\documentclass{article}

\usepackage{microtype}
\usepackage{graphicx}
\usepackage{subfigure}
\usepackage{booktabs}
\usepackage{color}
\usepackage{xcolor}
\usepackage{tcolorbox}
\usepackage{pifont}
\usepackage{xspace}
\usepackage{bbm}
\usepackage{tikz}
\usepackage{enumerate}
\usepackage{hyperref}
\usepackage{natbib}

\usepackage[accepted]{icml2025}

\usepackage{amsmath}
\usepackage{amssymb}
\usepackage{mathtools}
\usepackage{amsthm}

\usepackage[capitalize,noabbrev]{cleveref}

\DeclareMathOperator*{\argmax}{arg\,max}
\DeclareMathOperator*{\argmin}{arg\,min}
\newcommand{\globalprop}{\textsc{GlobalProp}\xspace}
\newcommand{\userprop}{\textsc{UserProp}\xspace}
\newcommand{\usereq}{\textsc{UserEQ}\xspace}
\newcommand{\scaledUP}{\textsc{ScaledUserProp}\xspace}
\newcommand{\half}{\frac{1}{2}}
\newcommand{\third}{\frac{1}{3}}
\newcommand{\I}{\mathcal{I}}
\newcommand{\Chat}{\widehat{C}}
\newcommand{\Max}{\textsc{Max}\xspace}
\newcommand{\Min}{\textsc{Min}\xspace}
\newcommand{\Med}{\textsc{Med}\xspace}
\newcommand{\Geo}{\textsc{Geo}\xspace}
\newcommand{\Util}{\textsc{Util}\xspace}
\newcommand{\Egal}{\textsc{Egal}\xspace}
\DeclareMathOperator*{\med}{med}
\let\vec\mathbf

\DeclarePairedDelimiter\abs{\lvert}{\rvert}%
\DeclarePairedDelimiter\norm{\lVert}{\rVert}%
\DeclarePairedDelimiter{\ceil}{\lceil}{\rceil}
\DeclarePairedDelimiter{\floor}{\lfloor}{\rfloor}
\newcommand{\icmlAlphabeticalOrdering}{\textsuperscript{*}Alphabetical ordering }

\definecolor{myred}{RGB}{201, 22, 22}
\definecolor{mygreen}{RGB}{38, 150, 68}
\newcommand{\yes}{\textcolor{mygreen}{\ding{51}}}
\newcommand{\no}{\textcolor{myred}{\ding{55}}}

\tcbset{size=small,top=1mm,bottom=1mm,middle=1mm}

\theoremstyle{plain}
\newtheorem{theorem}{Theorem}[section]
\newtheorem{proposition}[theorem]{Proposition}
\newtheorem{lemma}[theorem]{Lemma}

\theoremstyle{definition}
\newtheorem{definition}[theorem]{Definition}

\theoremstyle{remark}

\usepackage[textsize=tiny]{todonotes}

\icmltitlerunning{Fraud-Proof Revenue Division on Subscription Platforms}

\begin{document}

\twocolumn[
\icmltitle{Fraud-Proof Revenue Division on Subscription Platforms}

\icmlsetsymbol{alphabetical}{*}

\begin{icmlauthorlist}
\icmlauthor{Abheek Ghosh}{oxford,alphabetical}
\icmlauthor{Tzeh Yuan Neoh}{harvard,alphabetical}
\icmlauthor{Nicholas Teh}{oxford,alphabetical}
\icmlauthor{Giannis Tyrovolas}{oxford,alphabetical}
\end{icmlauthorlist}

\icmlaffiliation{oxford}{University of Oxford, UK}
\icmlaffiliation{harvard}{Harvard University, USA}

\icmlcorrespondingauthor{Tzeh Yuan Neoh}{tzehyuan\_neoh@harvard.edu}
\icmlcorrespondingauthor{Nicholas Teh}{nicholas.teh@cs.ox.ac.uk}
\icmlcorrespondingauthor{Giannis Tyrovolas}{giannis.tyrovolas@cs.ox.ac.uk}

\icmlkeywords{fraud-proof, revenue division, subscription platforms}

\vskip 0.3in
]

\printAffiliationsAndNotice{\icmlAlphabeticalOrdering}

\begin{abstract}
We study a model of subscription-based platforms where users pay a fixed fee for unlimited access to content, and creators receive a share of the revenue. Existing approaches to detecting fraud predominantly rely on machine learning methods, engaging in an ongoing arms race with bad actors. We explore revenue division mechanisms that inherently disincentivize manipulation. We formalize three types of manipulation-resistance axioms and examine which existing rules satisfy these. We show that a mechanism widely used by streaming platforms, not only fails to prevent fraud, but also makes detecting manipulation computationally intractable. We also introduce a novel rule, \scaledUP, that satisfies all three manipulation-resistance axioms. Finally, experiments with both real-world and synthetic streaming data support \scaledUP as a fairer alternative compared to existing rules.
\end{abstract}

\section{Introduction}
In September 2024, the FBI criminally charged a musician, Michael Smith, for orchestrating a scheme to fraudulently inflate his music streams on platforms such as Amazon Music, Apple Music, Spotify, and YouTube Music---and according to court documents, walked away with over US\$10~million in royalty payments \cite{us-indictment}.
To successfully execute his scheme, he utilized hundreds of thousands of songs created using AI, and built a complicated network of over a thousand bot accounts that artificially boost streams across these platforms billions of times. Although each stream originated from a bona fide, fee-paying account, the way the platform(s) distributed subscription revenue  allowed each bot to generate more in royalties than it cost to maintain its subscription.%

Subscription platforms have seen significant growth in recent years, driven by the rise of internet streaming services such as Spotify, Apple Music, Netflix, etc.
For instance, the annual revenue of the music streaming industry reached US\$27.6 billion in 2023, with significant increases over the last ten years \cite{ifpi-report}.
Under this business model, users pay a fixed subscription fee to enjoy unlimited access to all content on the platform, typically by content creators.
The platform then takes a fixed revenue cut and distributes the rest to the creators based on engagement metrics (e.g., play counts or views) and/or specific agreements between creators and platforms. %

Despite efforts to curb manipulation, bad actors persist, using bots and click-farms to inflate user engagement \cite{drott2020fakestreams,snickars2018spotibot}. This issue is so significant that major music streaming platforms like Amazon Music and Spotify have established an industry advocacy group \cite{musicfightsfraud} to combat such fraud, which is estimated to cost the industry US\$300 million annually \cite{burton2021playola}.
Additionally, the rise of AI-generated content introduces new challenges--- platforms are increasingly flooded with synthetic tracks, videos, and live streams designed to exploit engagement-driven algorithms.
This AI-generated content often amplifies fraudulent listening activities, making manipulation harder to detect.

Current machine learning (ML) approaches to this problem predominantly focus on \emph{detecting} fraudulent activity---using sophisticated algorithms ranging from anomaly detection \cite{esmaeilzadeh2022frauddetection} to unsupervised learning \cite{mollaoglu2021fraudunsupervised} and graph neural networks \cite{li2021livestreaming}.
For instance, music streaming platforms such as Spotify have proprietary models that identify whether a stream is legitimate (using meta-data such as IP location, listening patterns, and other information) and issue fines if they deem too many streams to be fraudulent \cite{spotify-globalprop}.

However, as AI continues to evolve, so do the methods used by fraudsters, leading to a continuous arms race. These bad actors increasingly leverage advanced automation tools to make fraudulent activities more sophisticated and harder to detect, challenging the robustness of existing detection frameworks and driving the need for innovative, adaptive solutions \cite{us-indictment}.

The root of the problem stems from the way revenue is currently distributed to content creators on most subscription-based streaming platforms: ``funds from the royalty pool are allocated proportionally among artists based on their respective percentages of total streams'' \cite{us-indictment}---we call this rule \globalprop.

In this paper, we tackle this problem from a \emph{mechanism design} perspective, i.e., we mathematically formalize notions of fraud in this setting and investigate the existence of revenue division mechanisms that can inherently disincentivize fraudulent behavior, %
thereby reducing the industry's reliance on expensive and complex fraud detection methods to combat manipulation. Moreover, if such mechanisms exist, they could complement existing ML-based approaches by providing a foundational layer of fraud resistance.
These mechanisms inherently target known forms of fraud, allowing ML systems to focus on adapting to emerging, previously unseen types of fraud that may arise in the future, ensuring continuous improvement in detecting and addressing manipulation. %

Additionally, many policymakers and academics have also argued against the fairness of \globalprop{} in favor of an alternative rule---\userprop{} (which directly allocates a fixed fraction of each user's subscription fee only among the creators of the content the user consumes)---from an economic \cite{meyn2023monetary,muikku2017comparative}, empirical \cite{moreau2024empirical}, theoretical \cite{bergantinos2023revenue}, and legal \cite{dimont2018switchlegal} perspective. Motivated by these debates, we aim to address fairness considerations in our work as well.

Lastly, the primary focus of our work is on fraudulent behavior specifically related to the creation of fake users (bots) to manipulate engagement metrics. We deliberately do not address the equally prevalent issue of widespread AI-generated content on these platforms. The legal status of such content can vary, especially since some popular artists openly release their AI-generated voices as (semi-)open-source \cite{josan2024ai}, making its permissibility platform-dependent and governed by specific rules.
Nonetheless, our work provides a principled framework for studying similar challenges. As AI continues to evolve and new forms of fraudulent behavior emerge, our approach can be extended to address these issues, provided that appropriate regulatory frameworks are established to guide the platforms. %

\subsection{Our Results}
In this work, we focus on designing \emph{manipulation-resistant} mechanisms from a computational and axiomatic perspective, setting our research apart from all previous work on this model.
Although we build on the standard model for subscription platforms established in prior literature, our key contribution lies in introducing several axioms that aim to capture both resistance to manipulation and maintaining fairness and analyzing these axioms with respect to multiple revenue-division mechanisms---three from existing literature and one novel mechanism that we propose.

Moreover, we challenge the current status quo rule, \globalprop, by demonstrating that \emph{detecting suspicious activity} under this rule is computationally intractable---an important finding in this context. Since \emph{fraud detection} (and fraud in general) is highly relevant to the ML community, we believe this result will be of particular interest to researchers and practitioners in the field.

In \Cref{sec:prelim}, we establish three fundamental properties that define the space of mechanisms we consider: \emph{anonymity}, \emph{neutrality}, and \emph{no free-ridership}. The first two ensure that payoffs to artists only depend on their engagement with users. In particular, mechanisms cannot distinguish between fraudulent and genuine artists or users.
No free-ridership eliminates trivial cases where an artist without engagement receives a non-zero payoff.
Next, we formalize three forms of resistance to strategic manipulation. \emph{Fraud-proofness} prevents adversaries from profitably creating new fraudulent users. \emph{Bribery-proofness} prevents profitably bribing existing users and is a strengthening of \emph{click-fraud-proofness} as presented in \citet{bergantinos2023revenue}.
Finally, \emph{(strong) Sybil‐proofness} ensures that artists cannot gain by splitting into multiple identities or merging with others.
All three axioms are novel in our setting and are motivated by real-world observations. %
We also introduce two additional fairness axioms---\emph{engagement monotonicity} and \emph{Pigou-Dalton consistency}, the latter inspired by an equitability concept in welfare economics.

In \Cref{sec:mechanisms}, we conduct an axiomatic study (with respect to our proposed concepts) of several rules proposed in the literature so far. Notably, we show that \globalprop fails to satisfy fraud-proofness and bribery-proofness, in contrast to the other two contenders---\userprop{} and \usereq{}.
Contributing to existing critiques of \globalprop{}, we establish a case against \globalprop{} through a computational lens, and in the context of fraud detection.
We show that if a platform uses \globalprop{}, detecting potentially fraudulent activity is NP-hard.
We then analyze the two other existing rules: \userprop{} and \usereq{}.
We study their axiomatic properties and prove that they satisfy our manipulation-resistance axioms, unlike \globalprop{}.
We also demonstrate that portioning rules cataloged in \citet{elkind2023portioning} fail all the manipulation-resistance axioms we consider.

Finally, in \Cref{sec:novelmech}, we propose and study a new rule---\scaledUP{}.
We show that it has the same axiomatic guarantees as \userprop{} but is fairer when measured by the popular ``pay-per-stream'' metric. We use this to quantify \emph{maximum envy} in this setting and empirically verify this against existing rules in \Cref{sec:experiments}.

All omitted proofs can be found in the paper's appendix.

\subsection{Related Work}

Our work considers the model proposed and studied by several recent works on (music) streaming platforms.\footnote{However, we note that this model is also applicable to many other content subscription platforms (e.g., education, art, etc.).}

\citet{alaei2022revenuesharing} and \citet{lei2023proratamusic} focused on a comparative study between \globalprop{} and \userprop{}.
More specifically, \citet{alaei2022revenuesharing} focused on providing characterizations of both rules with respect to newly proposed axioms. They were also concerned with which of these two rules could sustain a set of artists' profitability on the platform, as well as comparing them from both the platform's and the artists' perspectives.
\citet{lei2023proratamusic} pointed out the shortcomings of \userprop{}.
They compared the two rules primarily in terms of \emph{egalitarian fairness} (i.e., the lowest payout among all artists) and \emph{efficiency} (i.e., ``dominance on quality profile''), but they allow for artists to vary stream quality and thus this concept is not relevant in our model.

\citet{bergantinos2023revenue} go beyond previous works to consider a family of rules that interpolates between \globalprop and \userprop, and they provide further characterizations for both rules and their interpolation.
Subsequently, \citet{bergantinos2024shapley} introduced the Shapley index as a rule for this setting and characterized it using existing and new axioms.

\citet{deng2024computationalcopyright} investigate revenue‐sharing mechanisms for AI‐generated music platforms. Their work centers on the challenge of attributing a new, AI‐created track to specific copyrighted recordings in the training data---an attribution problem that underpins royalty allocation in that setting. This challenge is fundamentally distinct from the problems we address.

A related stream of work is the \emph{museum pass problem}, popular in the economics literature, and was first introduced by \citet{ginsburgh2001museum,ginsburgh2003museumpassgame}. The problem studies the sharing of revenue among museums from the sale of museum passes for a price below the aggregate
admission fee of individual member museums (i.e., bundled pricing).
\citet{beal2010museum} and \citet{ginsburgh2001museum,ginsburgh2003museumpassgame}  studied the problem as a \emph{coalitional game}, whereas \citet{casasmendez2011museumbankruptcy} and \citet{eestevez2012passepartout} studied the problem as a \emph{bankruptcy game}.
\citet{wang2011museumcost} studied the dual version of the problem---the museum cost sharing problem.
All of the works above (including several more recent works which look at the \emph{Shapley value} as a rule \cite{bergantinos2015museumaxiomatic,bergantinos2016museumpass}) essentially conduct an axiomatic study of popular rules in their respective games modeled, but adapted to this new setting.
We refer the reader to the  \citet{casasmendez2014museumsurvey} for a survey on earlier works on this area. From 2001 to 2014, works on the topic cumulatively studied more than $30$ axioms, with broadly two kinds of manipulation-resistant axioms---one based on ``ticket prices'' and the other based on ``reported number of visitors''.
However, we note that the museum pass problem is fundamentally different from our problem, and thus the way axioms (and rules) are conceptualized would also naturally be distinct. This distinction is particularly apparent when it comes to concepts relating to manipulation.

Our work also contributes to the broader literature on applying computational and algorithmic methods to address incentive-related challenges in online economic systems and platforms. For example, manipulation issues have been studied in the contexts of online advertising markets \citep{golrezaei2021contextualads, kanoria2014dynamicreserveprices}, recommendation systems \citep{eilat2023performativerecommendation, yao2023recsyscompeting}, and e-commerce platforms  \citep{golrezaei2021learningproductrankings,he2022marketfakereviews,mayzlin2014fakereviews}.

\section{Model and Axioms}\label{sec:prelim}
For each positive integer $k$, let $[k] := \{1,\dots,k\}$.
Let $N = [n]$ be the set of \emph{users} and $C = [m]$ be the set of \emph{artists}.
Suppose that an \emph{adversary} controls a set of \emph{fake users} $\widehat{N} \subseteq N$ and a set of \emph{fake artists} $\widehat{C} \subseteq C$; let $\widehat{n} = |\widehat{N}|$.
For each $i \in N$ and $j \in C$, let $w_{ij} \geq 0$ denote the number of \emph{interactions} user $i$ has with artist $j$.\footnote{This is typically defined as a \emph{stream} (on music streaming platforms like Spotify), whereby a user plays a track for a minimum duration, or a \emph{view} (on video streaming platforms like YouTube Live) when a user joins and stays for a minimum amount of time.} %
For each user $i \in N$, we assume
that $\sum_{j \in C} w_{ij} > 0$, i.e., the user has some non-zero interactions.\footnote{Note that in many of our proofs, we can without loss of generality assume that weights are rational numbers.} %
Let $\vec{w}_i = (w_{i1}, \dots, w_{im})$ for each $i \in N$.
The \emph{engagement profile} is  $\vec{w} = (\vec{w}_1,\dots,\vec{w}_n)$.

Without loss of generality, we assume that the subscription fee for each user is $1$ unit. %
Then, the total subscription fee collected from the users is $n$.
As assumed in the prior works on this topic, and as observed in the real-world, we assume that the platform takes a cut of $(1 - \alpha) n$ and distributes the remaining $\alpha n$ to the artists, for some $\alpha \in (0,1]$. %

A problem \emph{instance} $\I=(N, C, \vec{w})$ is defined by the set of users $N$, the set of artists $C$, and the engagement profile $\vec{w}$.
A \emph{payment rule} (or simply \emph{rule}) is a function $\phi$ that maps each instance $\I$ to an $m$-valued vector $ (\phi_{\I}(1), \dots, \phi_{\I}(m))$, where $\phi_{\I}(j)$ is the payment to artist $j \in C$.
To simplify notation, for a subset of artists $S \subseteq C$, we use $\phi_{\I}(S)$ to denote the sum of the payments to the artists in the set $S$: $\phi_{\I}(S) = \sum_{j \in S} \phi_{\I}(j)$.

\subsection{Preliminary Axioms} \label{sec:preliminaryaxioms}
We begin by introducing three fundamental fairness properties that any reasonable revenue division mechanism in our setting should satisfy.
We will then provide a rationale for the necessity of these axioms in our setting.

The first axiom---anonymity---prescribes that the rule cannot distinguish between real and fake users.
\begin{definition}[Anonymity]
    A rule $\phi$ is \emph{anonymous} if permuting the labels of the users does not affect the payoffs of the artists.
    Formally, rule $\phi$ is anonymous if for all instances $\I = (N, C, \vec{w})$ and $\I' = (N, C, \vec{w}')$ and all permutations $\sigma: N \rightarrow N$, if $\vec{w}_i = \vec{w}_{\sigma(i)}'$ for all users $i \in N$, then for all artists $j \in C$, $\phi_\I(j) = \phi_{\I'}(j)$.
\end{definition}

The second axiom---neutrality---is similar in nature to anonymity, but for artists. It prescribes that the rule cannot distinguish between real and fake artists.

\begin{definition}[Neutrality]
    A rule $\phi$ is \emph{neutral} if permuting the labels of the artists permutes their payoffs.
    Formally, rule $\phi$ is neutral if for all instances $\I = (N, C, \vec{w})$ and $\I' = (N, C, \vec{w}')$ and all permutations $\sigma: C \rightarrow C$, if $w_{ij} = w_{i\sigma(j)}'$ for all users $i \in N$ and artists $j \in C$, then for all artists $j \in C$, $\phi_\I(j) = \phi_{\I'}(\sigma(j))$.
\end{definition}
In our setting, it is crucial to consider only rules that are anonymous and neutral.
In practice, given the number of users/artists, it is virtually impossible to detect all fake users/artists, even with existing fraud detection techniques, as noted in our introduction. This inability to reliably distinguish between real and fake users or artists underscores the importance of addressing the questions we aim to answer.

Finally, the last fundamental axiom we consider is the notion of \emph{no free-ridership}.
Intuitively, this means that artists who receive no user engagement should not receive any payment.
\begin{definition}[No free-ridership]
    A rule $\phi$ satisfies \emph{no free-ridership} if, for any instance $\I = (N, C, \mathbf{w})$ and artist $j \in C$ where $\sum_{i \in N} w_{ij} = 0 $, then $\phi_\I(j) = 0$.
\end{definition}
This axiom rules out trivial rules that allocate payments disregarding user engagement (e.g., giving equal payment to each artist irrespective of user engagement) and are, therefore, resistant to strategic manipulation.

\subsection{Axioms for Preventing Strategic Manipulation}
We start by formalizing the fraud alleged in the indictment mentioned in the introduction.
Intuitively, no adversary should be able to create fake users ($\widehat{N}$), pay their subscription fee, and earn a profit from her own fake artists ($\widehat{C}$).\footnote{Note that we do not impose any constraints on the listening behavior or engagement profiles of these fake users.}
Rules that make such fraud impossible are \emph{fraud-proof}.

\begin{definition}[Fraud-proofness]
    A rule $\phi$ is  \emph{fraud-proof} if the following holds:
    For any two instances
    $\I = (N \setminus \widehat{N}, C, \vec{w})$ and $\I' = (N, C, \vec{w}')$ with $\mathbf{w}_i = \mathbf{w}'_i$ for all $i \in N \setminus \widehat{N}$, and any $\widehat{C} \subseteq C$, we have that $\phi_{\I'}(\widehat{C}) -  \phi_{\I}(\widehat{C}) \leq \widehat{n}$.

    A rule $\phi$ is \emph{single-user fraud-proof} if $\widehat{n} = 1$.
\end{definition}
Our definition of fraud-proofness considers only an adversary's profit from creating fake users, not fake artists.
This means an adversary can introduce fake artists to earn profits without using fake users. However, without fake users, any fake artist must attract engagement from real users to profit (by the no free-ridership assumption).
Whether this practice violates a platform’s rules is a separate issue beyond our scope---we focus on the extra profit an adversary can gain by adding fake users, assuming a fixed set of artists (which may include fake ones).

Next, we show that single-user fraud-proofness is equivalent to (multi-user) fraud-proofness, simplifying how one can reason about fraud-proofness.

\begin{proposition} \label{prop:singleuserFP}
    A rule $\phi$ is fraud-proof if and only if it is single-user fraud-proof.
\end{proposition}

Another form of manipulation is \emph{bribery}.
Bribery is particularly relevant in scenarios where the platform imposes substantially stringent access requirements, making creating fake users significantly more challenging.
However, under such conditions, artists may resort to colluding with and \emph{bribing} users---offering to pay the subscription fees of the users to manipulate their engagement profiles.
This practice is commonly observed in \emph{streaming farms}, the streaming equivalent of \emph{click farms} in advertising \cite{drott2020fakestreams}.
We call resistance to such bribery as \emph{bribery-proofness}. %
\begin{definition}[Bribery-proofness]
    A rule $\phi$ is \emph{bribery-proof} if the following holds:
    For any two instances $\I = (N, C, \mathbf{w})$ and $\I' = (N, C, \mathbf{w}')$ with $\vec{w}_i \neq \vec{w}_i'$ for exactly $k$ users, and any $\widehat{C} \subseteq C$, we have that $\phi_{{\I'}}(\widehat{C}) - \phi_{{\I}}({\widehat{C}}) \leq k$.

    A rule $\phi$ is \emph{single-user bribery-proof} if $k = 1$. %
\end{definition}

Similarly to fraud-proofness, multi-user bribery-proofness and single-user bribery-proofness are equivalent.

\begin{proposition} \label{prop:singleuserBP}
    A rule is bribery-proof if and only if it is single-user bribery-proof.
\end{proposition}

We note that (single-user) bribery-proofness substantially strengthens the axiom of \emph{click-fraud-proofness} proposed in \citet{bergantinos2023revenue}.
Click-fraud-proofness requires that a single user altering their engagement cannot alter the payoff of any artist by more than $1$. Formally, for all $j$, $\abs{\phi_{\I'}(j) - \phi_\I(j)} \leq 1$.
Single-user bribery-proofness requires that for all subsets of artists $\Chat \subseteq C$, $\abs{\phi_{\I'}(\Chat) - \phi_\I(\Chat)} \leq 1$.\footnote{Note that by \Cref{prop:singleuserBP}, it suffices to only consider single-user bribery-proofness.}
Bribery-proofness implies click-fraud-proofness and protects from multiple artists colluding.

Fraud-proofness and bribery-proofness capture resilience to two different kinds of manipulation.
Despite being similar, we show that the axioms are not equivalent.
Recall that $\alpha$ is the fraction of each user's subscription fee that is allocated to the artists, with the remaining portion retained by the platform as a fixed cut.
\begin{theorem} \label{thm:fp_bp_relationship}
    Consider some rule $\phi$. Then:
    \begin{enumerate}[(i)]
        \item If $\alpha = 1$ and $\phi$ is fraud-proof, it is also bribery-proof;
        \item For $\alpha \in (0,1]$, there exists a rule that is bribery-proof but not fraud-proof, even when $m=2$;
        \item For $\alpha < 1$, there exists a rule that is fraud-proof but not bribery-proof, even when $m=2$.
    \end{enumerate}
\end{theorem}

The last pair of axioms that we consider---\emph{Sybil-proofness}\footnote{The name is inspired by the concept of a \emph{Sybil attack} in computer networks.} and its strong counterpart---addresses a different form of strategic manipulation compared to the two earlier concepts.
Intuitively, these axioms are designed to prevent any artist(s) from splitting or merging to gain an unfair advantage and fraudulently increasing their revenue share, thus also ensuring that all artists are treated fairly based on their actual level of user engagement.

\begin{definition}[Sybil-proofness]
    A rule $\phi$ is \emph{Sybil-proof} if the following holds: For any two instances $\I = (N, C, \mathbf{w})$ and $\I' = (N, C', \mathbf{w}')$ whereby $C \subseteq C'$, if for every subset of artists $C^* \subseteq C$ such that
    \begin{enumerate}[(i)]
        \item $w_{ij} = w'_{ij}$ for all $i \in N, j \in C^*$; and
        \item $\sum_{j \in C \setminus C^*} w_{ij} = \sum_{j \in C' \setminus C^*} w'_{ij}$ for all $i \in N$,
    \end{enumerate}
    then we must have that $\phi_\I(C\setminus C^*) = \phi_{\I'}(C' \setminus C^*)$.
\end{definition}

We can define a stronger notion of Sybil-proofness by relaxing (i) and (ii), defined as follows. %
Note that strong Sybil-proofness implies Sybil-proofness.
\begin{definition}[Strong Sybil-proofness]
    A rule $\phi$ is \emph{strongly Sybil-proof} if the following holds:
    For any two instances $\I = (N, C, \mathbf{w})$ and $\I' = (N, C', \mathbf{w}')$ whereby $C \subseteq C'$, if for any subset of artists $C^* \subseteq C$ such that
    \begin{enumerate}[(i)]
        \item $\sum_{i \in N} w_{ij} = \sum_{i \in N} w'_{ij}$ for all $j \in C^*$; and
        \item $\sum_{i \in N} \sum_{j \in C \setminus C^*} w_{ij} = \sum_{i \in N} \sum_{j \in C' \setminus C^*} w'_{ij}$,
    \end{enumerate}
    then we must have that $\phi_\I(C\setminus C^*) = \phi_{\I'}(C' \setminus C^*)$.
\end{definition}

We will show later that \globalprop{} is the only neutral rule satisfying strong Sybil-proofness (\Cref{thm:strong_sybil_proofness_characterisation}), hence also motivating our study of (the weaker) Sybil-proofness.

\subsection{Fairness Axioms}
Next, we consider two fairness properties---engagement monotonicity and Pigou-Dalton consistency.

Intuitively, if an artist's engagement increases while every other artists' engagement does not increase, this artist's payoff should not decrease---this aligns with basic economic principles. %
It would be fundamentally unfair for a creator’s rising popularity to result in a lower payoff. We formalize this fairness property as follows.

\begin{definition}[Engagement monotonicity]
    A rule $\phi$ is \emph{engagement monotone} if the following holds:
    For any two instances $\I = (N, C, \mathbf{w})$ and $\I' = (N, C, \mathbf{w}')$, if there exists a $j^* \in C$ such that
    \begin{enumerate}[(i)]
        \item $w_{ij^*} \leq w'_{ij^*}$ for all $i \in N$; and
        \item $w_{ij} \geq w'_{ij}$ for all $i \in N$ and $j \in C \setminus \{j^*\}$,
    \end{enumerate}
    then we must have that $\phi_{{\I}}(j^*) \leq \phi_{{\I'}}({j^*})$.
\end{definition}

Next, the \emph{Pigou-Dalton principle} \cite{pigou1920welfare,dalton1920welfare}, is a fundamental fairness notion from welfare economics and often referenced in collective decision-making \cite{Moulin03}---it states that among similar outcomes, the equitable one should be picked.
We interpret this principle in our setting: all other things being equal, an artist who is more ``uniformly enjoyed'' should receive weakly more payoff from an equally popular but ``polarizing'' artist.

\begin{definition}[Pigou-Dalton consistency]
A rule $\phi$ is \emph{Pigou-Dalton consistent} if the following holds:
For any two instances $\I = (N, C, \mathbf{w})$ and $\I' = (N, C, \mathbf{w}')$, if there exists some $i,i' \in N$ and $j \in C$ such that
\begin{enumerate}[(i)]
    \item $w'_{ij} = w_{ij} - \delta$ (where $\delta > 0$ and $w_{ij} - \delta > 0$);
    \item $w'_{i'j} = w_{i'j} + \delta$ and $w'_{i'j} \leq w'_{ij}$; and
    \item $w_{kj'} = w'_{kj'}$ for all $k \in N$ and $j' \in C \setminus \{j\}$, and $w_{kj} = w'_{kj}$ for all $k \in N \setminus \{i,i'\}$.
\end{enumerate}
then we must have that $\phi_\I(j) \leq \phi_{\I'}(j)$.
\end{definition}

\section{Existing Mechanisms} \label{sec:mechanisms}
In this section, we formally define the three existing mechanisms proposed in the literature, and study which axioms they satisfy.
We summarize our results in \Cref{tab:properties}. At the end of the section, we also include a reference to a discussion on how our model generalizes \emph{portioning} rules.
\begin{table}
 \centering
\begin{tabular}{c || c c c c c}
 Axioms / Rules & \textsc{GP} & \textsc{UP} & \textsc{UEq} & \textsc{ScUP}\\
 \hline\hline
 Fraud-proofness & \no & \yes & \yes & \yes \\
 Bribery-proofness & \no & \yes & \yes & \yes \\
 Sybil-proofness & \yes & \yes & \no & \yes\\
 Strong Sybil-proofness & \yes & \no & \no & \no \\
 \hline
 Engagement monotonicity & \yes & \yes & \yes & \yes\\
 Pigou-Dalton consistency & \yes & \no & \yes & \no \\
\end{tabular}
\caption{Axiomatic properties of the revenue division mechanisms. \textsc{GP} is  \globalprop, \textsc{UP} is \userprop, \textsc{UEq} is \usereq, and \textsc{ScUP} is \scaledUP.} \label{tab:properties}
\end{table}

The rules we consider in this and the next section trivially satisfy anonymity and neutrality.
Therefore, among the three preliminary axioms introduced in \Cref{sec:preliminaryaxioms}, we will only formally prove the satisfaction of no free-ridership.

\subsection{\globalprop: The Status Quo}
\globalprop{} distributes the payoff to each artist proportionally to the artist's share of total engagement.
For example, if there are $500$ users, and an artist gets $25\%$ of the total user engagement in the platform, then the artist correspondingly receives a payment of $0.25 \times 500 \alpha = 125\alpha$ under \globalprop{}.
According to court documents \cite{us-indictment}, this is the rule that major streaming platforms use.\footnote{It is also sometimes known as the \emph{pro-rata} rule.}

\begin{tcolorbox}[title=\globalprop]
    Given an instance $\I = (N, C, \mathbf{w})$ and for each $j \in C$, the payment rule \globalprop{} is defined as follows.
    \tcblower
    \begin{equation*}
        \phi_{\I}(j) = \frac{\sum_{i \in N} w_{ij}}{\sum_{j'\in C}\sum_{i \in N} w_{ij'}} \times \alpha n.
    \end{equation*}
\end{tcolorbox}

It is easy to observe that users with higher engagement exert a disproportionate influence on revenue distribution.
Given this, it is not surprising that this rule fails to satisfy both fraud-proofness and bribery-proofness.

\begin{theorem} \label{thm:rules_globalprop_strategy}
    \globalprop{} satisfies strong Sybil-proofness, but fails fraud-proofness and bribery-proofness.
\end{theorem}

Moreover, strong Sybil-proofness uniquely characterizes \globalprop{}, given our neutrality assumption.
\begin{theorem}\label{thm:strong_sybil_proofness_characterisation}
    \globalprop{} is the only neutral rule satisfying strong Sybil-proofness.
\end{theorem}

\globalprop{} also satisfies our fairness axioms.
\begin{theorem} \label{thm:rules_globalprop_fairness}
    \globalprop{} satisfies no free-ridership, engagement monotonicity, and Pigou-Dalton consistency.
\end{theorem}

\paragraph{A Case Against \globalprop: The Computational Intractability of Fraud Detection.}
We have shown that \globalprop{} is not fraud-proof. One might hope that artists benefiting from fraud could be easily identified and removed. Unfortunately, detecting the artists who gain the most from fraudulent activity is computationally intractable.

Importantly, a user who streams music extensively is not inherently suspicious---some people naturally listen to music for most of their waking hours. Thus, instead of targeting individual active users, we should focus on identifying artists who may be used as vehicles for fraud by an adversary.\footnote{Our objective is to identify fraudulent artists as a means of detecting suspicious interactions between fake users and fake artists.}

\begin{definition}[Potentially Suspicious Profits]\label{def:psp}
    Given a set of artists $U \subseteq C$, their \emph{potentially suspicious profit (\textsc{PSP})} from \globalprop{} is their maximum marginal profits from a set of users $V$, less the cost of creating these users:
    \begin{align*}
        \textsc{PSP}&(U) = \max_{V \subseteq N} \left( \frac{\sum_{i \in N} \sum_{j \in U} w_{ij}}{\sum_{i \in N} \sum_{j \in C} w_{ij}} \times \alpha n \right. \\
        & \left. - \frac{\sum_{i \in N\setminus V} \sum_{j \in U} w_{ij}}{\sum_{i \in N\setminus V} \sum_{j \in C} w_{ij}} \times \alpha (n -|V|)  - |V| \right).
    \end{align*}
\end{definition}
Thus, our objective of identifying suspicious artists can be framed as finding a set of artists $U \subseteq C$ such that $\textsc{PSP}(U)$ is high.
However, the choice of $|U|$ is crucial.
If we restrict $U$ to a single artist ($|U|=1$), an adversary can easily evade detection by distributing fake users' listening activity across multiple fraudulent artists.
On the other hand, if we impose no constraint on $|U|$, we risk identifying a set of legitimate artists with dedicated fan bases. Also, while an adversary can create multiple fake artists, doing so incurs administrative overhead---such as setting up identification and banking details---which makes the creation of an arbitrarily large number of fake artists impractical in many circumstances.

Therefore, we define the problem of finding suspicious artists as finding the set $U \subseteq C$ of size at most $k$ artists that maximize $\textsc{PSP}(U)$. However, we show that this problem is computationally intractable, with the following result.

\begin{theorem}\label{thm: NP-c}
    Given an instance $\I = (N,C, \mathbf{w})$ and parameters $k \le |C|$ and $\gamma > 0$, it is NP-hard to determine if there exists a $U \subseteq C$ such that $|U| \leq k$ and $\textsc{PSP}(U) \geq \gamma$.
\end{theorem}

\subsection{User-Additive Rules}
\label{sec:user-additive-rules}
At the opposite extreme from \globalprop are rules where each user's subscription fee is distributed solely based on their individual engagement profile. Under these rules, an artist's total payoff is simply the sum of the amounts they would receive from each user in a single-user setting.
We refer to this class of rules as \emph{user-additive}.\footnote{This term is distinct from \emph{user-centric}, which is sometimes used in the literature to refer to \userprop{}.}
\begin{definition}[User-additive rules]
    For each instance $\I = (N, C, \vec{w})$, define instances $\I_i = (\{i\}, C, \vec{w}_i)$ for each $i \in N$.
    Then, a rule $\phi$ is \emph{user-additive} if for all instances $\I$ and artists $j \in C$,
    $\phi_\I(j) = \sum_{i \in N} \phi_{\I_i}(j)$.
\end{definition}
We then show the following.
\begin{proposition}
    \label{thm:user-additive-implies-monotone-fp-bp}
    A user-additive rule is fraud-proof and bribery-proof.
\end{proposition}

We focus on two user-additive rules that have been discussed in the existing literature: \userprop and \usereq.
Under \userprop, an $\alpha$ fraction of each user's subscription fee is allocated to the artists proportional to the user's engagement. For example, if a user listens to three artists---the first artist $50\%$ of the time and the other two artists $25\%$ each---then under \userprop{}, the artists will receive payments of $\alpha/2$, $\alpha/4$, and $\alpha/4$ from this user, respectively. The total payment of an artist is the sum of such payments from each user.

\begin{tcolorbox}[title=\userprop]
    Given an instance $\I = (N, C, \mathbf{w})$ and for each $j \in C$, the payment rule \userprop{} is defined as follows.
    \tcblower
    \begin{equation*}
        \phi_{\I}(j) = \sum_{i \in N} \frac{w_{ij}}{\sum_{j' \in C} w_{ij'}} \times \alpha.
    \end{equation*}
\end{tcolorbox}
We show that it satisfies all of the manipulation-resistant axioms (excluding strong Sybil-proofness) and engagement monotonicity, but fails Pigou-Dalton consistency.
\begin{theorem} \label{thm:rules_userprop_strategy}
    \userprop{} is fraud-proof, bribery-proof, and Sybil-proof, but fails strong Sybil-proofness.
\end{theorem}

\begin{theorem} \label{thm:rules_userprop_fairness}
    \userprop{} satisfies no free-ridership and engagement monotonicity, but fails Pigou-Dalton consistency.
\end{theorem}

Next, we consider the \usereq{} rule, first studied in \citet{bergantinos2024shapley}.
They established the equivalence between \usereq{} and the \emph{Shapley value}, a fundamental measure in cooperative game theory that ensures a fair distribution of payoffs among players based on their contributions \cite{shapley1953}.

Now, given an instance $\I = (N, C, \mathbf{w})$, for each $i \in N$ and $j \in C$, let $\mathbf{1}_{w_{ij}>0}$ be the indicator function that returns the value $1$ if $w_{ij} > 0$, and $0$ otherwise.
In \usereq, an $\alpha$ fraction of each user's subscription fee is distributed equally among the artists with strictly positive engagement from the user.
For example, if a user listens to only three artists---$80\%$, $19\%$, and $1\%$ of the time, respectively---and does not listen to other artists, then these three artists each receives a payment of $\alpha/3$ from this user, and the remaining artists do not receive any payment from the user. The total payment to an artist is the sum of such payments from each user.
\begin{tcolorbox}[title=\usereq]
    Given an instance $\I = (N, C, \mathbf{w})$ and for each $j \in C$, the payment rule \usereq{} is defined as follows.
    \tcblower
    \begin{equation*}
        \phi_{\I}(j) = \sum_{i \in N} \frac{\mathbf{1}_{w_{ij}>0}}{|\{j' \in C : w_{ij'} > 0\}|} \times \alpha.
    \end{equation*}
\end{tcolorbox}
\usereq{} has similar guarantees as \userprop{}, with the difference being that it fails Sybil-proofness, but satisfies Pigou-Dalton consistency.
\begin{theorem}
\label{thm:rules_usereq_strategy}
    \usereq{} is fraud-proof and bribery-proof, but  fails Sybil-proofness.
\end{theorem}

\begin{theorem} \label{thm:rules_usereq_fairness}
    \usereq{} satisfies no free-ridership, engagement monotonicity, and Pigou-Dalton consistency.
\end{theorem}

\subsection*{A Generalization of Portioning}

We also make an important observation: our model can be viewed as a generalization of \emph{portioning} under cardinal preferences \cite{elkind2023portioning,freeman2021truthfulbudget},\footnote{We refer the reader to a recent survey by \citet{SuTe2026VotingDivisible} on works in this area.} where each agent subjectively divides a contiguous resource (such as time or money) among a given set of \emph{candidates}, and the goal is to aggregate these preferences to obtain one (fair) division.
This is similar to our model if we let agents be users, candidates be artists, and preferences be interactions.\footnote{Note that this requires imposing rational number constraints on preferences, as assumed in the preliminaries.}
However, portioning rules require that the engagement of each user is normalized (i.e., sums to $1$).
We can then generate rules for our setting by normalizing each $\vec{w}_i$ and applying a portioning rule to the instance.
There are eight portioning rules cataloged in \citet{elkind2023portioning}.
One of them is equivalent to \userprop{}, but the other seven fail fraud-proofness, bribery-proofness and Sybil-proofness.
We present these rules and prove the results in \Cref{sec:portioning-appendix}.

\section{\scaledUP: A Fairer Mechanism} \label{sec:novelmech}
The three rules we considered above are conceptually distinct: \globalprop allows dedicated fans to exert a disproportionate influence on revenue distribution, but this also creates opportunities for fraud by fabricating users who may \emph{appear} as dedicated fans.
In contrast, \userprop{} is often viewed by policymakers as a more desirable alternative to \globalprop{}. However, \userprop is not necessarily fairer \cite{lei2023proratamusic}, and user-additive rules in general may fail to meaningfully reward artists for increasing the engagement within their existing fanbase.

To better understand differences in \emph{payment fairness}, it is useful to examine the \emph{pay-per-stream} metric \cite{dimont2018switchlegal,meyn2023monetary}.
Given an instance $\I$ and an artist $j$, let the artist \emph{pay-per-stream} (\textsc{PPS}) for rule $\phi$ be $\textsc{PPS}(\phi, \I ,j) = \frac{\phi_{\I}(j)}{\sum_{i \in N} w_{ij}}$.
Using this, we define the \emph{maximum envy} (\textsc{ME}) of $\I$ as $\textsc{ME}(\phi, \I) = \frac{\max_{j \in C} \textsc{PPS}(\phi, \I, j) }{\min_{j' \in C} \textsc{PPS}(\phi, \I, j') }$.
This ratio quantifies the disparity in \textsc{PPS} between the highest-paid and lowest-paid artists, providing a measure of the maximum envy in revenue distribution.

Then, we obtain the following result, which essentially implies that any fraud-proof or bribery-proof rule has the potential to be extremely unfair (unbounded maximum envy).
\begin{proposition} \label{prop:pps_fpbp}
    For all $\alpha \in (0,1]$ and rules $\phi$, if there exists $k \in \mathbb{R}$ such that for all instances $\I$, $\textsc{ME}(\phi, \I) \leq k$, then $\phi$ fails fraud-proofness and bribery-proofness.
\end{proposition}

However, not all such rules may perform equally bad on this front---we will analyze this later through experiments (in \Cref{sec:experiments}), with a slight variant of the \textsc{ME} definition.

Given this, we attempt to achieve a compromise by designing a rule that has the same axiomatic guarantees as \userprop{}, while offering empirically (in \Cref{sec:experiments}) stronger fairness guarantees than \userprop{} and \usereq{}.
$\scaledUP$ works by having the platform take a \emph{disproportionate amount} of commission from low-engagement users.
The platform then runs \userprop on the remaining subscription fees. It is defined as follows.

\begin{tcolorbox}[title=\scaledUP]
    Given an instance $\I = (N, C, \mathbf{w})$, let $\gamma$ be a constant such that $\sum_{i \in N}  \min\left(\gamma \cdot \sum_{j \in C}w_{ij}, 1\right)= \alpha n$. Then, for each $j \in C$, the payment rule \scaledUP{} is defined as follows.
    \tcblower
    \vspace{-0.3cm}
    \begin{equation*}
        \phi_\I(j) = \sum_{i \in N} \left( \min(\gamma \cdot \sum_{j' \in C} w_{ij'}, 1) \times \frac{w_{ij}}{\sum_{j' \in C} w_{ij'}} \right).
    \end{equation*}
\end{tcolorbox}

Note that when $\alpha = 1$, we have $\min(\gamma \cdot \sum_{j' \in C} w_{ij'}, 1) = 1$ for all $i \in N$, making \scaledUP equivalent to \userprop.
For $\alpha < 1$, if no user's engagement exceeds $\frac{1}{\alpha}$ times the average engagement, then \scaledUP is equivalent to \globalprop, which we show below.

\begin{theorem}
    \label{thm:scaled-up-similar-to-gp}
    Fix an instance $\I = (N, C, \mathbf{w})$. If $\sum_{j \in C} w_{ij} \leq \frac{1}{n\alpha} \sum_{i \in N} \sum_{j \in C} w_{ij}$ for all $i \in N$, then \scaledUP is equivalent to \globalprop.
\end{theorem}
Thus, \scaledUP can be viewed as a variant of \globalprop that ``limits the influence'' of users who have engagement significantly above average.
We then show that \scaledUP has exactly the same axiomatic  guarantees as \userprop{}, with the following results.

\begin{theorem}\label{thm:rules_scaled_strategy}
    \scaledUP satisfies fraud-proofness, bribery-proofness, and Sybil-proofness, but fails strong Sybil-proofness.
\end{theorem}

\begin{theorem} \label{thm:rules_scaled_fairness}
    \scaledUP satisfies no free-ridership, engagement monotonicity, but fails Pigou-Dalton consistency.
\end{theorem}

\section{Experiments} \label{sec:experiments}
To complement our theoretical analysis, we conduct experiments to evaluate our fraud-proof (and bribery-proof) mechanisms---\userprop{}, \usereq{}, \scaledUP{}---using both synthetic and real-world datasets. Motivated by our definition of \emph{maximum envy} in \Cref{prop:pps_fpbp}, for each rule, we analyze the top and bottom few artists based on their \emph{pay-per-stream (\textsc{PPS})} relative to \globalprop{}'s \textsc{PPS}, as the revenue share ($\alpha$) varies.\footnote{Note that in \Cref{prop:pps_fpbp}, maximum envy is defined with respect to the single top and bottom user, which differs from the metric used in this section.
In our experiments, we chose to report metrics for the top and bottom few users rather than just the
single best and worst, as we believe this provides a more robust assessment—mitigating the impact of
potential outliers that may disproportionately affect the extremes.
However, our definition and theoretical results would easily extend to top and bottom few users, making it consistent with that used for the experiments.} Notably, only $\scaledUP$ is influenced non-linearly by changes in $\alpha$ (the other rules scale linearly with $\alpha$). Consequently, the pay-per-stream values for the other three rules remain constant across different values of $\alpha$.

\textbf{Synthetic datasets} \quad We generate synthetic problem instances involving $10,000$ users and $1,000$ artists.
For each user, we first determine the number of artists they interact with by drawing a value uniformly at random from the range $[1, 100]$.
Based on this value, we randomly select the corresponding number of artists from the pool of $1,000$. For each chosen artist, the number of times the user streams their music is sampled from a Poisson distribution with \(\lambda = 1\). We repeat the experiments $100$ times.

\textbf{Real-world datasets} \quad We utilize data from the \emph{Music Listening Histories Dataset} \cite{vigliensoni17music}, that contains the listening history of approximately $583,000$ users, $439,000$ artists, and a cumulative total of $27$ billion \emph{listening events} (i.e., user-artist interactions).\footnote{Our code is accessible at \url{https://github.com/nicteh/Fraud-Proof-Revenue-Division}.} %

\textbf{Discussion} \quad
\begin{figure}[t]
    \centering
    \subfigure[Real data, top $100$ artists' \textsc{PPS} relative to \textsc{GP}]{%
        \includegraphics[width=0.48\columnwidth]{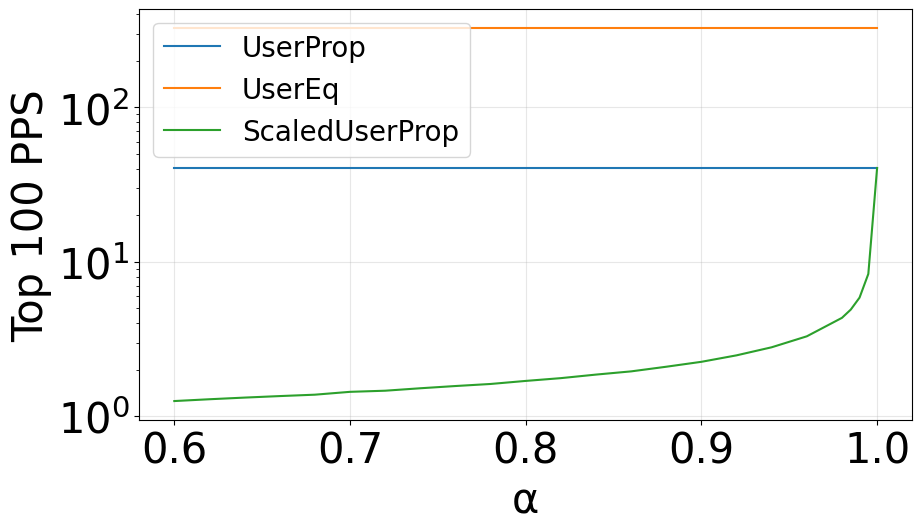}
    }%
    \hfill
    \subfigure[Real data, bottom $100$ artists' \textsc{PPS} relative to \textsc{GP}]{%
        \includegraphics[width=0.48\columnwidth]{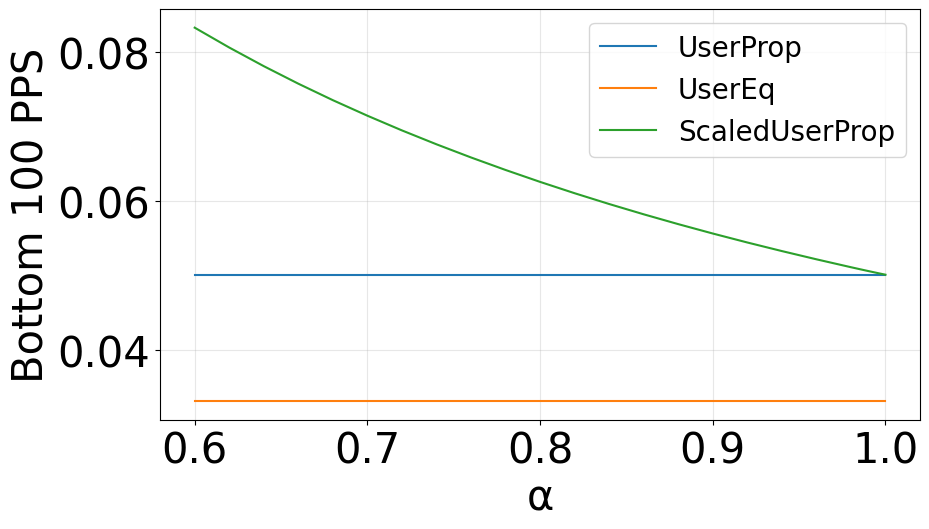}
    }%
    \vskip\baselineskip
    \subfigure[Synthetic data, top $10$ artists' \textsc{PPS} relative to \textsc{GP}]{%
        \includegraphics[width=0.48\columnwidth]{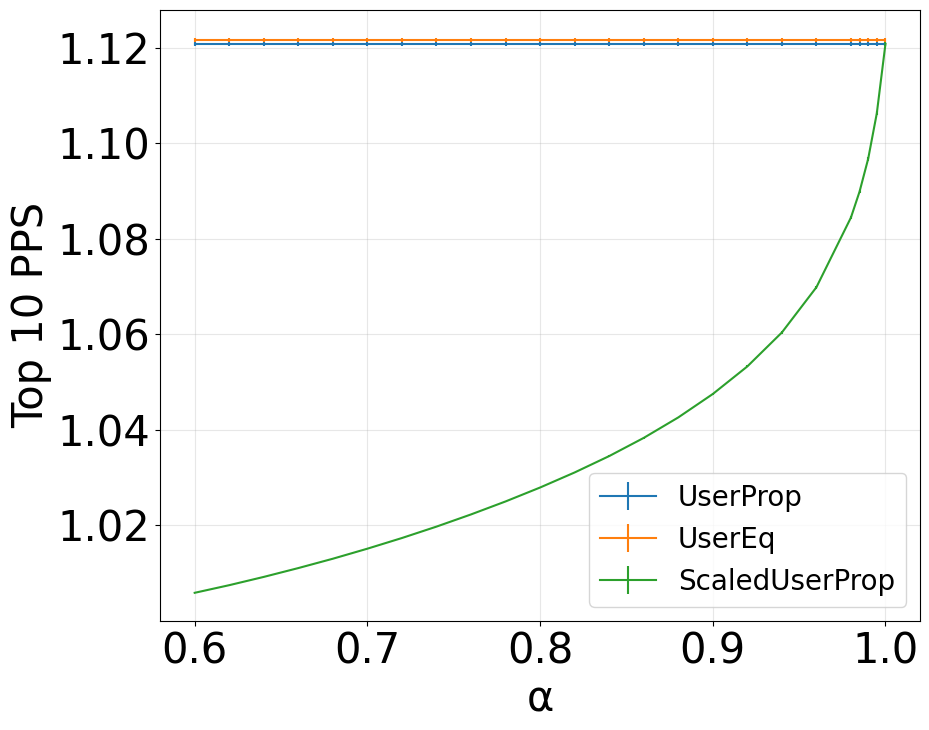}
    }%
    \hfill
    \subfigure[Synthetic data, bottom $10$ artists' \textsc{PPS} relative to \textsc{GP}]{%
        \includegraphics[width=0.48\columnwidth]{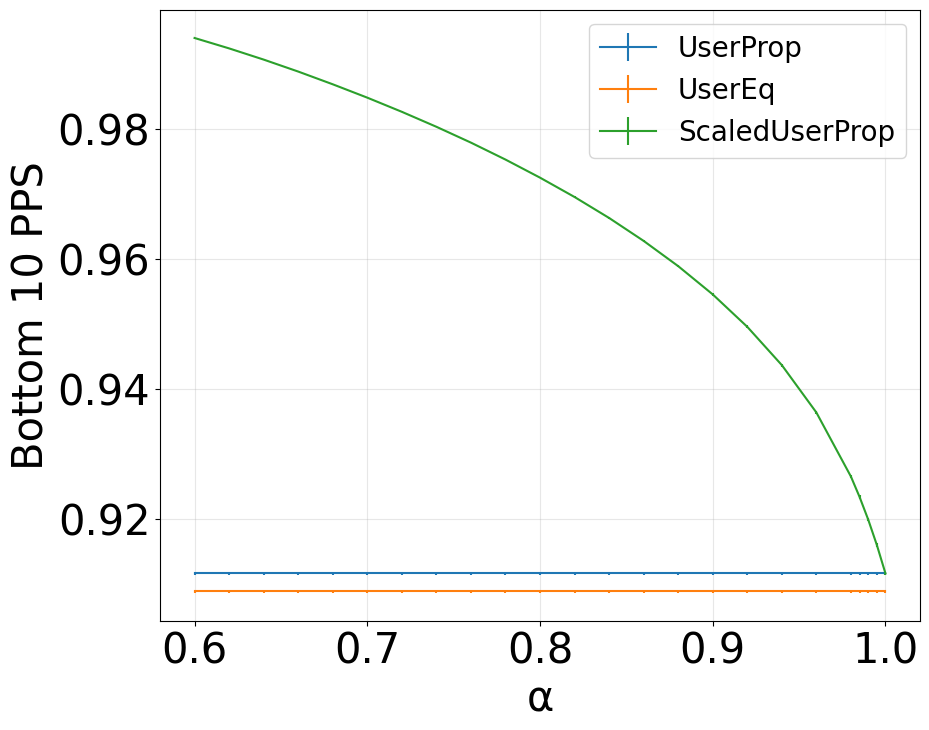}
    }%
    \caption{Overview of graphs from real and synthetic data. (a) and (b) show results for real data, while (c) and (d) show results for synthetic data. \textsc{GP} is short for \globalprop.}
    \label{fig:graphs}
\end{figure}
On real-world data, \scaledUP{} emerges as fairest mechanism among those considered, especially for values of \(\alpha\) not close to $1$; whereas \usereq{}, which treats avid and casual listeners equally, is the least fair.
\scaledUP{} significantly reduces the top $100$ artists' \textsc{PPS} even for \(\alpha > 0.9\), but it only gradually increases the bottom $100$ \textsc{PPS} as \(\alpha\) decreases.
To understand this outcome, we first observe that artists with high \textsc{PPS} typically attract infrequent listeners, while those with low \textsc{PPS} tend to have a more dedicated, avid fanbase.

We also observe that under \scaledUP, each stream from a user contributes \(\min(\gamma, \frac{1}{\sum_{j \in C} w_{ij}})\), whereas under \userprop{}, it contributes \(\frac{\alpha}{\sum_{j \in C} w_{ij}}\).
For avid listeners with high \(\sum_{j \in C} w_{ij}\), a stream under \scaledUP is worth $\frac{1}{\alpha}$ times its value under \userprop.
Conversely, for infrequent listeners, \scaledUP{} caps a stream's worth at \(\gamma\), while under \userprop, it can reach up to \(\alpha\) in the extreme case where \(\sum_{j \in C} w_{ij} = 1\).

On synthetic data, \scaledUP{} remains the fairest mechanism as \(\alpha\) decreases.
However, in contrast to the real-world data, we observe two key differences:
(1) the top and bottom \textsc{PPS} are much closer in magnitude, and (2) \userprop and \usereq perform nearly identically.
These differences can be partly attributed to the way synthetic instances are generated.
While our model accounts for users with varying streaming frequencies, it does not capture the real-world tendency of certain artists to attract predominantly avid or infrequent listeners.

\section{Conclusion}

In this work, we formalized three types of manipulation by fraudulent agents in subscription-based platforms, motivated by a real-world multi-million dollar fraud case.
We show that the axioms we introduced are not equivalent and study the rules that satisfy them.
\globalprop{}, which is used by streaming platforms, does not satisfy fraud-proofness or bribery-proofness.
However, we show that \userprop{} and \usereq do.
We introduce a novel rule, \scaledUP. It is as strong in resisting manipulation as \userprop{} and incentivizes artists to increase their overall engagement similarly to \globalprop{}.
Our empirical study on real and synthetic data of fraud-proof rules support \scaledUP is a fairer fraud-proof alternative to other rules.

A natural follow-up direction would be to study a \emph{freemium} model, by incorporating users who have to watch advertisements to gain access to content on the platform, and have been adopted by streaming platforms such as YouTube and Spotify, among others. Revenue division in this context would have different considerations and call for more appropriate axioms to be defined. Machine learning approaches have been adopted here as well \cite{goli2024personalads}; it would be interesting to explore these questions from a mechanism design perspective.

\bibliography{icml25}
\bibliographystyle{icml2025}

\newpage
\appendix
\onecolumn

\begin{center}
\LARGE\bf
Appendix\end{center}

\section{Omitted Proofs from Section \ref{sec:prelim}}
\subsection{Proof of \Cref{prop:singleuserFP}}
If $\phi$ is fraud-proof then by definition it is single-user fraud-proof.
    Now, suppose rule $\phi$ is single-user fraud-proof.
    Consider instances $\I = (N, C, \vec{w})$ and $\I' = (N \cup \widehat{N}, C, \vec{w}')$ with $N \cap \widehat{N} = \varnothing$ and let $\widehat{C} \subseteq C$.
    Enumerate $\widehat{N} = \{\widehat{n}_1, \ldots, \widehat{n}_k\}$, then for $j \leq k$ we construct instances $\I_j = (N \cup \{\widehat{n}_1, \ldots, \widehat{n}_j\}, C, \vec{w} \mid \vec{w}_{\widehat{n}_1}, \ldots, \vec{w}_{\widehat{n}_j})$ where we adjoin engagement vectors $\vec{w}_{\widehat{n}_1}, \ldots, \vec{w}_{\widehat{n}_j}$ to $\vec{w}$.
    We have $\I_0 = \I$ and $\I_{n_k} = \I'$.

    By single user fraud-proofness, for all $j$: $\phi_{\I_{j+1}}(\widehat{C}) - \phi_{\I_j}(\widehat{C}) \leq 1$.
    So, $\sum_{j = 0}^{k-1} \phi_{\I_{j+1}}(\widehat{C}) - \phi_{\I_{j}}(\widehat{C}) \leq k$, but as a telescoping sum, $\phi_{\I_{k}}(\widehat{C}) - \phi_{\I_{0}}(\widehat{C}) = \phi_{\I'}(\widehat{C}) - \phi_{\I}(\widehat{C}) \leq k$.
    So, $\phi$ is fraud-proof.

\subsection{Proof of \Cref{prop:singleuserBP}}
    If a rule is bribery-proof it is also by definition single-user bribery proof.
    Suppose a rule is not bribery-proof.
    Then, there are instances $\I$, $\I'$ with $\vec{w}_i \neq \vec{w}_i'$ precisely for users $\{1, \ldots, k\}$ and $C^+ \subseteq C$ with $\phi_{{\I'}}({C^+}) -  \phi_{{\I}}({C^+}) > k$.
    Now, consider instances $\I_0 = \I, \I_1, \ldots, \I_k = \I'$ with the profile of user $i$ in instance $\I_j$ being $\vec{w}_i'$ if $i \leq j$ and $\vec{w}_i$ otherwise.
    Then $\sum_{j = 0}^{k-1} \phi_{\I_{j+1}}({C^{+}}) - \phi_{\I_{j}}(C^{+}) = \phi_{\I'}(C^+) -  \phi_{\I}(C^+) > k$ and so in particular at least one term in the sum is greater than $1$.
    So the rule is not single-user bribery-proof.

\subsection{Proof of \Cref{thm:fp_bp_relationship}}

\paragraph{(i)} Suppose rule $\phi$ is not bribery-proof and consider a pair of instances $\I$ and $\I'$ such that bribery-proofness is violated.
    Let $C^+$ the set of artists with a higher payoff in $\I'$, namely $C^+ = \{c \mid \phi_{\I'}(c) > \phi_{\I}(c)\}$.
    We similarly define $C^= = \{c \mid \phi_{\I'}(c) = \phi_{\I}(c)\}$ and $C^- = \{c \mid \phi_{\I'}(c) < \phi_{\I}(c)\}$.
    Since $\phi$ violates bribery-proofness on $\I$ and $\I'$, $\phi_{\I'}(C^+)-\phi_{\I}(C^+) > 1$.

    Now, consider an instance with one less user: $\mathcal{F}$.
    As $\alpha = 1$, $\phi_{\mathcal{F}}(C) = \phi_{\I}(C) - 1$.
    By fraud-proofness, $\phi_{\mathcal{F}}({C^+ \cup C^=}) \geq \phi_{\I'}({C^+ \cup C^=}) - 1$ and $\phi_{\mathcal{F}}({C^-}) \geq \phi_{\I}({C^-}) - 1$. So, adding up the inequalities, $\phi_{\mathcal{F}}(C) \geq \phi_{\I'}({C^+ \cup C^=}) + \phi_{\I}({C^-}) - 2$. As this is a bribery-proofness violation, $\phi_{\I'}({C^+}) > \phi_{\I}({C^+}) + 1$, and by definition $\phi_{\I}({C^=}) = \phi_{\I'}({C^=})$. So, $\phi_{\mathcal{F}}(C) > \phi_{\I}({C}) - 1$, but $\phi_{\mathcal{F}}(C) = \phi_{\I}(C) - 1$, giving rise to a contradiction.

\paragraph{(ii)}
    We define a rule that is bribery-proof but not fraud-proof.
    This rule is anonymous, neutral and satisfies no free-ridership.
    To do so, we will modify the rule \userprop{} which is defined in \cref{sec:user-additive-rules}.
    We will set a threshold value of $\beta = 2 \floor*{\frac{n\alpha}{20}}$.
    Consider an instance $\I = (N, \{0, 1\}, \vec{w})$ with 2 artists.
    For each artist $j \in \{0, 1\}$, we compute $p_j = \sum_{i \in N}\frac{w_{ij}}{w_{i, j} + w_{i, 1 - j}} \alpha$.
    If $\min(p_0, p_1) \geq 2 \floor*{\frac{n\alpha}{20}}$, for $j \in \{0, 1\}$,
    $\phi_\I(j) = p_j$.
    Otherwise, let $j$ the artist with $p_j < p_{1-j}$. Let the number of users that have positive engagement with artist $j$ be $a_j$.
    Then, $\phi(j) = \min(a_j, \beta)$ and $\phi(1-j) = n\alpha - \min(a_j, \beta)$.

    This rule is bribery-proof. Suppose we have a bribery-proofness violation in instances $\I = (N, \{0, 1\}, \vec{w}), \I'= (N, \{0, 1\}, \vec{w}')$. Let $k$ the unique user that modifies her engagement profile and $\phi_{\I'}(j) > \phi_\I(j) + 1$.
    We define $p_j'$ for instance $\I'$ analogously with $p_j$, $p_j' =\sum_{i \in N}\frac{w_{ij}'}{w_{i, j}' + w_{i, 1 - j}'} \alpha$.
    Then, since $w_{ij} = w_{ij}'$ for all $i \neq k$, we have that $p_j' - p_j = \frac{w_{kj}'}{w_{k, j}' + w_{k, 1 - j}'} \alpha - \frac{w_{kj}}{w_{k, j} + w_{k, 1 - j}} \alpha \leq \alpha \leq 1$.
    Also, notice that $\beta$ is equal in both instances as $n$ is unchanged.

    We proceed by a case analysis. If $\phi_{\I}(j) = p_j$, then $\phi_{\I'}(j) \leq \max(p_j, p_j') \leq p_j + 1$. So, there can be no bribery-proofness violation if $\min(p_0, p_1) \geq \beta$.
    Suppose instead that for artist $j$, $p_j < \beta$. Notice that the number of users engaging with artist $j$, $a_j$, is greater than $p_j$. So $\phi_{\I}(j) \geq p_j$.
    Also, note that since exactly one user modifies her engagement between $\I$ and $\I'$, $a_j' \leq a_j + 1$.
    If $\phi_{\I}(j) = a_j$, then $\phi_{\I'}(j) \leq \max(p_j', a_j')$.
    But we have that $p_j' \leq p_j + 1 \leq a_j+1$ and $a_j' \leq a_j+1$. So, $\phi_{\I'}(j) \leq a_j+1 = \phi_{\I}(j) + 1$.
    Suppose instead that $\phi_{\I}(j) = \beta$. Then, $\phi_{\I'}(j) > \phi_{\I}(j)$ implies that $\phi_{\I'}(j) = p_j'$. But $p_j \leq \beta$ and $p_j' \leq p_j+1$, so $\phi_{\I'}(j) \leq \phi_{\I}(j) + 1$. This concludes the proof.

    However, this rule is not fraud-proof. For any $\alpha$, take $n = \ceil*{\frac{40}{\alpha} - 1}$. Then, $\beta = 2\floor*{\frac{n\alpha}{20}} = 2$. Construct an instance with $n$ users and where each user's engagement is $\vec{w}_i = (0.01, 0.99)$. Then, $p_0 = 0.01 \cdot  n\alpha \leq \beta = 2$. So $\phi_{\I}(j) = 2$. Suppose we add an extra user with profile $(0.01, 0.99)$. Then $\beta = 4$ as the number of users is now greater than $\frac{40}{a}$. As $a_0 = n > 4$, $\phi_{\I'}(j) = 4 > \phi_\I(j) + 1 = 3$. This constitutes a fraud-proofness violation.

    This rule is anonymous, neutral and satisfies no-freeridership. Anonymity and neutrality should be immediate. In the case that no user engages with an artist $j$ then $p_j = a_j = 0$ and so $\phi_\I(j) = 0$, satisfying no-freeridership.

\paragraph{(iii)}
    Suppose $\alpha < 1$, then we construct a rule $\phi$ that is fraud-proof but not bribery-proof. For ease of presentation, we add a surrogate rule $\psi$, that is then modified to make $\phi$ satisfy no-freeridership.
    Consider an instance with two candidates.
    Let $n$ the number of users and $\varepsilon$ a small positive constant, such that $ 0 < \varepsilon < 1 - \alpha$.
    Let $k$ the smallest integer such that $k\alpha > 2(1 + \varepsilon)$.
    For $n \leq k$, $\psi(j) = \frac{n\alpha}{2}$.
    For $n > k$, the rule distributes the payoff based on the number of users approving an artist.
    Let $a_j$ the number of users $i$ with $w_{ij} > 0$.
    If $a_0 = a_1$, then $\psi(0) = \psi(1) = \frac{n\alpha}{2}$.
    For $a_j> a_{1-j}$, then $\psi(j) = \frac{n\alpha + 1 + \varepsilon}{2}$ and $\psi(1 - j) = \frac{n\alpha + 1 - \varepsilon}{2}$.

    We construct $\phi$ using $\psi$. If for $j \in \{0, 1\}$, $\psi_\I(j) \leq a_j$ then $\phi_\I(j) = \psi_\I(j)$.
    If $a_j < \psi_\I(j)$ for some $j$ then $\phi_\I(j) = a_j$ and $\phi_\I(1 - j) = n\alpha - a_j$. Notice that since we disallow users with $(0, 0)$ engagement, $a_0 + a_{1} \geq n > n\alpha$. As $\psi_\I(0) + \psi_\I(1) = n\alpha$, $a_j < \psi_\I(j)$ implies that $\psi_\I(1 - j) < a_{1-j}$, so $\phi$ is well defined.
    Equationally, $\phi_\I(j) = \max(\min(\psi_\I(j), a_j), n\alpha - a_{1-j})$

    Now, to show that the rule is fraud-proof. Suppose there was a fraud-proofness violation with $\I = (N, \{0, 1\}, \vec{w}), \I' = (N \cup \{k\}, \{0, 1\}, \vec{w}')$ with $\vec{w}_i = \vec{w}_i'$ for all $i < k$. Let $j$ be the artist benefiting from fraud so $\phi_{\I'}(j) > \phi_{\I}(j) + 1$.

    First we prove that $\psi_{\I'}(j) \leq \psi_\I(j) + 1$.
    If $a_j < a_{1 - j}$ in $\I$ then we cannot have $a_j' > a_{1 - j}$ in $\I'$ as we add exactly one user.
    So, if $a_j < a_{1 - j}$ in $\I$, then $\psi_{\I'}(j) \leq \frac{(n+1)\alpha}{2}$. So, $\psi_{\I'}(j) - \psi_\I(j) \leq \frac{(n+1)a}{2} - \frac{n\alpha - 1 - \varepsilon}{2} = \frac{\alpha + 1 + \varepsilon}{2} < 1$ by our choice of $\varepsilon$. If $a_j \geq a_{1-j}$ then $\psi_\I(j) \geq \frac{n\alpha}{2}$ and $\psi_{\I'}(j) \leq \frac{n\alpha + 1 +\varepsilon}{2}$ and again $\psi_{\I'}(j) \leq \psi_\I(j) + 1$.

    Now, $\phi_\I(j) = \max(\min(\psi_\I(j), a_j), n\alpha - a_{1-j})$ and $\phi_{\I'}(j) = \max(\min(\psi_{\I'}(j), a_j'), (n+1)\alpha - a_{1-j}')$.
    We have proven that $\psi_{\I'}(j) \leq \psi_\I(j) + 1$. As we add one user, $a_j' \leq a_j + 1$. Finally, no agent is removed so $a_{1 - j}' \geq a_{1-j}$, so $(n+1)\alpha - a_{1-j}' \leq n\alpha - a_{1-j} + 1$.
    So, $\phi_{\I'}(j) \leq \max(\min(\psi_\I(j) + 1, a_j+ 1), n\alpha - a_{1-j}+1) = \max(\min(\psi_\I(j), a_j), n\alpha - a_{1-j}) + 1 = \phi_\I(j) + 1$ proving fraud-proofness.

    However, the rule is not bribery-proof. For a concrete example, let $\alpha = \frac{1}{2}$, $\varepsilon = \frac{1}{4}$. Let $\I =([5], \{0,1\}, \vec{w})$ with $\vec{w}_1=\vec{w}_2=\vec{w}_3=(1, 0)$ and $\vec{w}_4 = \vec{w}_5=(0,1)$.
    Then, $\phi_\I(1) = \frac{n\alpha + 1+ \varepsilon}{2}=\frac{5}{8}$.
    But, if we construct $\I'$ by setting $\vec{w}_3=(0, 1)$, then $\phi_{\I'}(1) = \frac{15}{8}$. But, $\phi_{\I'}(1) -\phi_{\I}(1) = \frac{10}{8} > 1$, which violates bribery-proofness.

\subsection{User-addition monotonicity}
As an additional tool, we consider the \emph{user-addition monotonicity} property, which will be frequently used in proving several of our axioms.
Intuitively, it states that adding a user should not decrease an artist's payoff.
This property is considerably strong and implies fraud-proofness and bribery-proofness. With user-addition monotonicity the axiom implications are captured by \cref{fig:relationship}.
\begin{figure}[h]
    \centering
    \begin{tikzpicture}[scale=1]

        \node (FP) at (2,2) {Fraud-proofness};
        \node (BP) at (-2,2) {Bribery-proofness};
        \node (Mon) at (0,1) {User-addition monotonicity};

        \draw[->] (Mon) -- (FP);
        \draw[->] (Mon) -- (BP);
        \draw[dashed, ->] (FP) -- (BP) node[midway, above] {$\alpha = 1$};

        \draw[rounded corners] (-4.1,0.5) rectangle (4.1,2.5);
    \end{tikzpicture}
    \caption{Relationship of axioms, arrows denote implications. The dashed arrow denotes conditional implication. }

    \label{fig:relationship}
\end{figure}
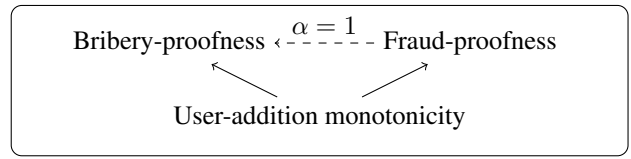

\begin{definition}[User-addition monotonicity]
    For an instance $\I$ and any engagement profile $\vec{w}_{n+1}$ consider instance $\I^{n+1}$ constructed by adding a user with profile $\vec{w}_{n+1}$ to $\I$.
    A rule $\phi$ satisfies \emph{user-addition monotonicity} if for all $\I$, $\vec{w}_{n+1}$ and $\I^{n+1}$ for all artists $c \in C$, $\phi_{\I}(c) \leq \phi_{\I^{n+1}}(c)$.
\end{definition}

\begin{proposition}
    \label{prop:user-addition-monotone-implies-fp-bp}
    If a rule is user-addition monotone, then it is both fraud-proof and bribery-proof.
\end{proposition}
\begin{proof}
    Consider instances $\I$ and $\I^{n+1}$ for some $\vec{w}_{n+1}$.
    A user adds $\alpha$ to the total payoff: $\phi_{\I^{n+1}}(C) - \phi_\I(C) = \alpha$. So for
    $\Chat \subseteq C$:
    \begin{equation*} \phi_{\I^{n+1}}(\Chat) - \phi_{\I}(\Chat) + \phi_{\I^{n+1}}(C \setminus \Chat) - \phi_\I(C \setminus \Chat) = \alpha
    \end{equation*}
    But, by monotonicity, for $S \subseteq C$ the marginal contribution of user $n+1$ is non-negative: $\phi_{\I^{n+1}}(S) - \phi_{\I}(S) \geq 0$.
    So, $\phi_{\I^{n+1}}(\Chat) - \phi_{\I}(\Chat) \leq \alpha \leq 1$
    and $\phi$ is fraud-proof.

    Now to prove bribery-proofness.
    Consider an instance $\I = (N, C, \vec{w})$. We construct instance $\I^{-n} = (N \setminus \{n\}, C, \vec{w}_{-n})$ with user $n$ removed. Take any instance $\I' = (N,C, \vec{w}')$ with $\vec{w}_{i} = \vec{w}_i'$ for all $i \neq n$.
    Then, for all $\Chat \subseteq C$, $\phi_{\I^{-n}}(\Chat) - \phi_\I(\Chat) \leq 0$ by monotonicity.
    By fraud-proofness, $\phi_{\I'}(\Chat) - \phi_{\I^{-n}}(\Chat) \leq 1$. Adding up, for all $\I$ and $\I'$ with engagement differing for a single user $\phi_{\I'}(\Chat) - \phi_{\I}(\Chat) \leq 1$, proving bribery-proofness.
\end{proof}

\section{Omitted Proofs from \Cref{sec:mechanisms}}

\subsection{Proof of \Cref{thm:rules_globalprop_strategy}}
We will prove each property separately.

\paragraph{\globalprop is not fraud-proof.}
Consider an instance $\I = (N, \{1, 2\}, \vec{w})$ with $\abs{N} > \frac{2}{\alpha} + 1$.
Let $\vec{w}_i = (1, 0)$ for all $i \in N$, so that $\phi_\I(2) = 0$.
Then, constructing an instance $\I'$ by adding a single profile $\vec{w}_{n+1} = (0, n)$, would result in a payoff of $\phi_{\I'}(2) = \frac{n}{2n} (n+1) \alpha > 1$ by assumption on $n$, contradicting fraud-proofness.

\paragraph{\globalprop is not bribery-proof.}
Similarly, for an instance $\I = (N, \{1, 2\}, \vec{w})$ with $\abs{N} > \frac{2}{\alpha} + 1$ and for each $i$, $\vec{w}_i = (1, 0)$ we have that $\phi_\I(2) = 0$.
However, if we construct $\I'$ by bribing user $n$ to change their profile to $\vec{w}_n' = (0, n)$, $\phi_{\I'}(2) = \frac{n}{2n} n\alpha > 1$ by assumption.

\subsection{Proof of \Cref{thm:strong_sybil_proofness_characterisation}}
Suppose $\phi$ is strongly Sybil-proof and neutral.

    Observe first, that if $\phi$ is strongly Sybil-proof, there exists a function $f$ such that:
    \[
        \phi_\I(c) = f\left(\sum_{i \in N} w_{ic},\sum_{i \in N} \sum_{j \in C} w_{ij},N\right)
    \]

    To see this, suppose there are instances $\I = (N, C, \mathbf{w})$ and $\I' = (N, C', \mathbf{w}')$ with
    $\sum_{i \in N} w_{ic} =\sum_{i \in N} w'_{ic}$ and $\sum_{i \in N} \sum_{j \in C'} w_{ij} = \sum_{i \in N} \sum_{j \in C} w'_{ij}$.
    So, $\sum_{i \in N} \sum_{j \neq c} w_{ij} = \sum_{i \in N} \sum_{j \neq c} w'_{ij}$ and
    the criteria for strong Sybil-proofness hold for $C^* = \{c\}$.
    So, $\phi_\I(C \setminus \{c\}) = \phi_{\I'}(C' \setminus \{c\})$. Because the number of users is equal in $\I$ and $\I'$, $\phi_\I(C) = \phi_{\I'}(C') = \abs{N}\alpha$.
    Hence,
        \begin{align*}
        \phi_\I(c) & = \phi_\I(C) - \phi_\I(C \setminus \{c\}) \\
        & = \phi_{\I'}(C') - \phi_{\I'}(C' \setminus \{c\}) \\
        & = \phi_{\I'}(c).
    \end{align*}

    We now claim that $f$ is a linear function of $\sum_{i \in N} w_{ic}$.\footnote{Here, we consider linearity as typically defined in linear algebra, and thus exclude affine functions.}
    To see this, observe that $f \left( \sum_{i \in N} w_{ic}, \sum_{i \in N} \sum_{j \in C} w_{ij},N \right) = \sum_{i \in N} w_{ic} \times g \left( \sum_{i \in N} \sum_{j \in C} w_{ij}, N \right)$.
    Clearly, if $\sum_{i \in N} w_{ic} = 0$, then for all $T$ and $N$, $f(0,T,N) = 0$.
    For any instance $\I = (N, C^* \cup\{c\}, \vec{w})$ with $c, d, e \notin C^*$ and $\beta \in (0, 1)$,  we construct $\I^\beta = (N, C^* \cup \{d, e\}, \vec{w}')$.
    For $j \notin \{d, e\}$, $w_{ij}' = w_{ij}$.
    We let $w_{id}'= \beta w_{ic}$ and $w_{ie}'= (1 - \beta) w_{ic}$.
    So, strong Sybil-proofness applies for $C^*$ and so $\phi_\I(c) = \phi_{\I^\beta}(d) + \phi_{\I^\beta}(e)$.

    But the total engagement of the users and the number of users is equal in $\I$ and $\I^\beta$. So, $f$ is linear on $\sum_{i \in N} w_{ij}$.
    Now, suppose we fix $\sum_{i \in N} \sum_{j \in C} w_{ij} = T$ and $N$.
    By linearity, if $\sum_{i \in N} w_{ic} = 0$ then $ f\left(0, T, N\right) = 0$.
    Conversely, if all artists other than $c$ receive $0$ engagement from all users, user $c$ will receive the entire payoff of $n\alpha$:
    $\phi_\I(c) = f\left(T, T, N\right) = n \alpha$. This determines $f$ uniquely:

    \[
        f\left(\sum_{i \in N} w_{ic},\sum_{i \in N} \sum_{j \in C} w_{ij},N\right) = \frac{\sum_{i \in N} w_{ic}}{\sum_{i \in N} \sum_{j \in C} w_{ij}}\times n \alpha.
    \]
    Which is equivalent to \globalprop{}.

\subsection{Proof of \Cref{thm:rules_globalprop_fairness}}
We will prove each property separately.

\paragraph{\globalprop satisfies no free-ridership.}
    Consider an instance $\I = (N, C, \mathbf{w})$.
    For every $j \in C$ where $\sum_{i \in N} w_{ij} = 0$,
    \begin{equation*}
        \phi_{\I}(j) = \frac{\sum_{i \in N} w_{ij}}{\sum_{j'\in C}\sum_{i \in N} w_{ij'}} \times \alpha n = 0,
    \end{equation*}
    since we assume $\sum_{j' \in C} w_{ij'} > 0$ for all $i \in N$.

\paragraph{\globalprop{} is engagement monotone.}
Consider any two instances $\I = (N, C, \mathbf{w})$ and $\I' = (N, C, \mathbf{w}')$ whereby for some $j^* \in C$, we have that (i) $w_{ij^*} \leq w'_{ij^*}$ for all $i \in N$, and (ii) $w_{ij} \geq w'_{ij}$ for all $i \in N$ and $j \in C \setminus \{j^*\}$.

    Now, since
    \begin{equation*}
        \sum_{i \in N} w'_{ij^*} \geq \sum_{i \in N} w_{ij^*}
        \quad \text{and} \quad
        \sum_{j \in C \setminus \{j^*\}} \sum_{i \in N} w_{ij} \geq \sum_{j \in C \setminus \{j^*\}} \sum_{i \in N} w'_{ij},
    \end{equation*}
    we get that
    \begin{equation*}
        \sum_{i \in N} w'_{ij^*} \cdot \sum_{j \in C \setminus \{j^*\}} \sum_{i \in N} w_{ij}
        \geq
        \sum_{i \in N} w_{ij^*} \cdot \sum_{j \in C \setminus \{j^*\}} \sum_{i \in N} w'_{ij}.
    \end{equation*}
    Adding $\sum_{i \in N} w'_{ij^*} \cdot \sum_{i \in N} w_{ij^*}$ to both sides of the equation, we can factorize the expressions on each side to obtain
    \begin{equation*}
        \sum_{i \in N} w'_{ij^*} \cdot \left( \sum_{j \in C \setminus \{j^*\}} \sum_{i \in N} w_{ij} + \sum_{i \in N} w_{ij^*} \right) \geq
        \sum_{i \in N} w_{ij^*} \cdot \left( \sum_{j \in C \setminus \{j^*\}} \sum_{i \in N} w'_{ij} + \sum_{i \in N} w'_{ij^*} \right).
    \end{equation*}
    Algebraic manipulation (note that by our model assumption, for each $i \in N$, $\sum_{j' \in C} w_{ij'} > 0$ and $\sum_{j' \in C} w'_{ij'} > 0$) gives us
    \begin{equation*}
        \frac{\sum_{i \in N} w_{ij^*}}{\sum_{j \in C \setminus \{j^*\}} \sum_{i \in N} w_{ij} + \sum_{i \in N} w_{ij^*}}
        \leq
        \frac{\sum_{i \in N} w'_{ij^*}}{\sum_{j \in C \setminus \{j^*\}} \sum_{i \in N} w'_{ij} + \sum_{i \in N} w'_{ij^*}},
    \end{equation*}
    which simplifies to
    \begin{equation*}
        \frac{\sum_{i \in N} w_{ij^*}}{\sum_{j' \in C} \sum_{i \in N} w_{ij'}}
        \leq
        \frac{\sum_{i \in N} w'_{ij^*}}{\sum_{j' \in C} \sum_{i \in N} w'_{ij'}}.
    \end{equation*}
    Consequently, we have that
    \begin{equation*}
        \phi_\I(j^*) = \frac{\sum_{i \in N} w_{ij^*}}{ \sum_{j' \in C} \sum_{i \in N} w_{ij'}} \times \alpha n
        \leq
        \frac{\sum_{i \in N} w'_{ij^*}}{ \sum_{j' \in C} \sum_{i \in N} w'_{ij'}} \times \alpha n = \phi_{\I'}(j^*). \qedhere
    \end{equation*}

    \paragraph{\globalprop{} is Pigou-Dalton consistent.}
    Consider any two instances $\I = (N, C, \mathbf{w})$ and $\I' = (N, C, \mathbf{w}')$ where there exists some $i,i' \in N$ and $j \in C$ such that
    \begin{enumerate}[(i)]
        \item $w'_{ij} = w_{ij} - \delta$ (where $\delta > 0$ and $w_{ij} - \delta > 0$);
        \item $w'_{i'j} = w_{i'j} + \delta$ and $w_{i'j} \leq w_{ij}$; and
        \item $w_{kj'} = w'_{kj'}$ for all $k \in N$ and $j' \in C \setminus \{j\}$, and $w_{kj} = w'_{kj}$ for all $k \in N \setminus \{i,i'\}$.
    \end{enumerate}
    Then, we get that
    \begin{align*}
        \phi_\I(j) & = \frac{\sum_{k \in N} w_{kj}}{\sum_{j' \in C} \sum_{k \in N} w_{kj'}} \times \alpha n\\
        & = \frac{w_{ij} + w_{i'j} + \sum_{k \in N \setminus \{i,i'\}} w_{kj}}{w_{ij} + w_{i'j} + \left(\sum_{j' \in C} \sum_{k \in N} w_{kj'} - w_{ij} - w_{i'j} \right)} \times \alpha n \\
        & = \frac{w'_{ij} + \delta + w'_{i'j} - \delta + \sum_{k \in N \setminus \{i,i'\}} w'_{kj}}{w'_{ij} + \delta + w'_{i'j} - \delta + \left(\sum_{j' \in C} \sum_{k \in N} w'_{kj'} - w'_{ij} - \delta - w'_{i'j} + \delta \right)} \times \alpha n \quad \text{(using (i), (ii), and (iii))}\\
        & = \frac{\sum_{k \in N} w'_{kj}}{\sum_{j' \in C} \sum_{k \in N} w'_{kj'}} \times \alpha n \\
        & = \phi_{\I'}(j),
    \end{align*}
    as desired.

\subsection{Proof of \Cref{thm: NP-c}}

We reduce from the Small Set Bipartite Vertex Expansion (SSBVE) problem. The SSBVE problem is known to be NP-complete and cannot be approximated better than $O(|V|^{1/4})$, where $V$ is the set of vertices, under plausible complexity conjectures \cite{chlamtavc2017minimizing}. We first define the neighborhood of a set of vertices in a graph and then formally define the decision variant of SSBVE.

\begin{definition}[Neighborhood]
    For a graph $G = (V, E)$ and a subset of vertices $S \subseteq V$, the neighborhood of $S$ is defined as $N(S) = \{ v \mid \exists (u,v) \in E \land u \in S\}$. Slightly overloading notation, let $N(v) = N(\{v\})$ for $v \in V$.
\end{definition}
\begin{definition}[Small Set Bipartite Vertex Expansion (SSBVE)]
    Given a bipartite graph $(U,V,E)$ and integers $\ell \leq |U|$ and $\delta \le |V|$, is there an $S \subseteq U$ with $|S| \geq \ell$ and $|N(S)| \leq \delta$?
\end{definition}

We are given an arbitrary instance of SSBVE: $(U, V, E, \ell, \delta)$. Let $U = \{u_1, \dots, u_{|U|})$ and $V = \{v_1, \dots v_{|V|}\}$. Let $d = \max_{u \in U} |N(u)|$ be the maximum number of neighbors, i.e., degree, of any vertex in $U$.

For our reduction, we will construct an instance $\I = (N, C, \vec{w})$ with $|N| = t + |U|$ users and $|C| = t + |V| + 1$ artists, where the value of $t$ is specified later, and with $\vec{w}$ defined as follows
\begin{equation*}
    w_{ij} = \begin{cases}
       \alpha d,  &\text{ if $i \in [t]$ and $j = i$}, \\
       1,  &\text{ if $i - t \in [|U|]$, $j - t \in [|V|]$ and $(u_{i - t} , v_{j-t}) \in E$}, \\
       d + 1 - |N(u_{i-t})|,  &\text{ if $i - t \in [|U|]$ and $j = t + |V| + 1$}, \\
       0, & \text{ otherwise}.
    \end{cases}
\end{equation*}
Intuitively, the first $t$ users and $t$ artists are dummies, where user $i$ listens to only artist $j = i$, and does so $\alpha d$ times,\footnote{In this proof, we allow the $w_{ij}$ values to be non-integers. As long as these are rational numbers, e.g., if $\alpha$ is a rational number, which is a reasonable assumption in practice, we could scale the weights to make everything integral.} but does not listen to any other artists.
The next $|U|$ users and the next $|V|$ artists correspond to the nodes in $U$ and $V$, respectively.
The final artist, artist $t + |V| + 1$, ensures that the total listening activity of each user $i \in N \setminus [t]$ is $d+1$, i.e., for all $i \in [t+1, \dots, t+|U|], \sum_{j \in C} w_{ij}  =  d+1$.

Note that each user streams at least $\alpha d$ times, i.e., $\forall i \in N, \sum_{j \in C} w_{ij}  \geq  \alpha d$. Thus, the pay-per-stream (\textsc{PPS}) of \globalprop{} for the instance will be $\textsc{PPS}(\I) \leq \frac{\alpha}{\alpha d} = \frac{1}{d}$. Note that the pay-per-stream remains bounded above by $\frac{1}{d}$ even if we remove some users from instance $\I$ as we maintain the property that each user streams at least $\alpha d$ times.

Next, we show that for all $\epsilon > 0$, as long as $t \geq \frac{(d+1)|U|}{\alpha d \epsilon}$, the pay-per-stream of $\I$ is at least $\frac{1}{d} - \epsilon$.
\begin{lemma}
    If $t \geq \frac{(d+1)|U|}{\alpha d \epsilon}$, then $\textsc{PPS}(\I) \geq \frac{1}{d} - \epsilon$.
\end{lemma}
\begin{proof}
    The pay-per-stream in $\I$ is $\textsc{PPS}(\I) = \frac{\alpha(t+|U|)}{\sum_{i \in N}\sum_{j \in C} w_{ij}} = \frac{\alpha(t+|U|)}{t\alpha d + |U| (d+1)}$. Thus,
    \begin{align*}
      \frac{1}{d} - \textsc{PPS}(\I) &=  \frac{1}{d} -\frac{\alpha(t+|U|)}{t\alpha d + |U| (d+1)} \\
      &= \frac{|U|(d+1)d  - \alpha |U| d }{t \alpha d^2 + d (d+1) |U|} \\
      &< \frac{|U|(d+1)}{t \alpha d }, \qquad\text{as $\alpha , d , |U|$ are all positive.}
    \end{align*}
If $t \geq \frac{(d+1)|U|}{\alpha d \epsilon}$, then $  \frac{1}{d} - \textsc{PPS}(\I) < \epsilon$, and thus  $\textsc{PPS}(\I) \geq \frac{1}{d} - \epsilon$, as required.
\end{proof}

We note that if $\epsilon < \frac{1}{d |U| ( d(\delta+1) + 1)}$,
then $\frac{\ell - 1}{d} = \frac{\ell}{d} - \frac{1}{d} < \frac{\ell}{d} - \epsilon |U| (d (\delta +1) + 1)$.
Furthermore, at this value of $\epsilon$, as $d \leq |V|$ and $\delta \leq |V|$, we have $t = \mathcal{O}(d^2 \delta |U|^2) = (|V|^3 |U|^2)$, and this reduction can be done in polynomial time.

We now prove that there is a $C' \subseteq C$ such that $|C'| \leq k = \delta + 1$ and $\textsc{PSP}(C') \ge \gamma = \frac{\ell - 1}{d}$ if and only if there is an $S \subseteq U$ with $|S| \geq \ell$ and $|N(S)| \leq \delta$.

Let $C' \subseteq C$ be the subset of artists that maximizes $\textsc{PSP}(C')$ among all subsets of size at most $\delta + 1$, i.e., $C' = \argmax_{\widehat{C} \in C, |\widehat{C}| \le \delta+1} \textsc{PSP}(\widehat{C}) $. Further, $\textsc{PSP}(C')$ is maximized using some subset of users as defined in \Cref{def:psp}; let $N' \subseteq N$ be the smallest among those subsets, i.e.,
\begin{align*}
    \mathcal{N} &= \argmax_{\widehat{N} \subseteq N} \left( \frac{\sum_{i \in N} \sum_{j \in C'} w_{ij}}{\sum_{i \in N} \sum_{j \in C} w_{ij}}  \alpha (t + |U|) - \frac{\sum_{i \in N\setminus\{ \widehat{N} \}} \sum_{j \in C'} w_{ij}}{\sum_{i \in N\setminus\{ \widehat{N} \}} \sum_{j\in C} w_{ij}}  \alpha (t + |U| - |\widehat{N}|)  - |\widehat{N}| \right), \\
    N' &= \argmin_{\widehat{N} \in \mathcal{N}} |\widehat{N}|.
\end{align*}

Next, we show that $N'$ does not contain any of the first $t$ users.
\begin{lemma}
    $[t] \cap N' = \emptyset$.
\end{lemma}
\begin{proof}
For the purpose of contradiction, let $[t] \cap N' \neq \emptyset$. Let us pick an $i' \in [t] \cap N'$.
Consider the three instances $\I_1$, $\I_2$, and $\I_3$ defined as follows:
\begin{itemize}
    \item $\I_1$ removes all users in $N'$ from $\I$.
    \item $\I_2$ removes all users in $N' \setminus \{i'\}$ from $\I$.
    \item $\I_3$ is constructed as follows: In the instance $\I_2$, for some $j \in C'$, increase $w_{i' j}$ until $\sum_{j \in C} w_{i'j} = \frac{\sum_{i \in N\setminus\{N'\}} \sum_{j \in C} w_{ij'}}{|N| - |N'|}$ (note that $\sum_{j \in C} w_{i'j}$ was originally $\alpha d$ because $i' \in [t]$, which is the minimum possible total engagement for any user, so we are in fact increasing $w_{ij'}$).
\end{itemize}
Notice that $\I_1$, $\I_2$, and $\I_3$ differ only with respect to user $i'$, where $\I_1$ does not contain $i'$, $\I_2$ contains $i'$ with its original engagement vector, while $\I_3$ contains $i'$ with an increased engagement for artist $j \in C'$ to ensure that the total engagement of user $i'$, and therefore, the average total engagement per user of $\I_3$ matches that of $\I_1$. As the engagement of $i'$ is exactly equal to the average engagement of users in $\I_3$ and $\I_1$, so $i'$ controls exactly $\frac{1}{|N| - |N'| + 1}$ fraction of the \globalprop{} allocation of $\I_3$. Therefore,
\[
    \phi_{I_3}(C') - \phi_{I_1}(C') \le \frac{1}{|N| - |N'| + 1}  \alpha (|N| - |N'| + 1) \le \alpha.
\]
Furthermore, as \globalprop{} is engagement monotone (\Cref{thm:rules_globalprop_fairness}), we have $\phi_{\I_3}(C') \geq \phi_{\I_2}(C')$. So,
\[
    \phi_{I_2}(C') - \phi_{I_1}(C') \leq \phi_{I_3}(C') - \phi_{I_1}(C') \le \alpha < 1.
\]
As the difference in the total payment to the artists in $C'$ from instances $I_1$ and $I_2$ is less than $1$, so the marginal profit of $N'$ is less than $N' \setminus \{ i \}$, which is a contradiction.
\end{proof}

Next, we show that $N'$ does not contain any user $i \in N \setminus [t]$, if $\sum_{j \in C'} w_{ij} \leq d$.

\begin{lemma}
    If $i \in [t+1, \dots, t+|U|] \cap N'$, then $\sum_{j \in C'} w_{ij} = d + 1$.
\end{lemma}
\begin{proof}
For the purpose of contradiction, let there be an $i' \in (N \setminus [t]) \cap N'$ such that $\sum_{j \in C'} w_{i'j} \leq d$. Let us consider the two instances $\I_1$ and $\I_2$ defined as: $\I_1$ removes all users in $N'$ from $\I$, and $\I_2$ removes all uses in $N' \setminus \{i'\}$ from $\I$.
As user $i'$ streams $d+1$ times, which is the maximum possible, we have $\textsc{PPS}(\I_2) \le \textsc{PPS}(\I_1)$. Further, as each user streams at least $\alpha d$ times, we have $\textsc{PPS}(\I_2) \le \frac{\alpha}{\alpha d} = \frac{1}{d}$. Using assumption $\sum_{j \in C'} w_{i'j} \leq d$, we have
\begin{align*}
    \phi_{I_2}(C') &- \phi_{I_1}(C') \\
    &= \textsc{PPS}(\I_2)  \left(\sum_{j \in C'} w_{i'j} + \sum_{i \in N\setminus\{N'\}} \sum_{j \in C'} w_{ij}\right) - \textsc{PPS}(\I_1)  \left(\sum_{i \in N\setminus\{N'\}} \sum_{j \in C'} w_{ij}\right) \\
    &\leq \textsc{PPS}(\I_2)  \left(d + \sum_{i \in N\setminus\{N'\}} \sum_{j \in C'} w_{ij}\right) - \textsc{PPS}(\I_1)  \left(\sum_{i \in N\setminus\{N'\}} \sum_{j \in C'} w_{ij}\right)\\
    & \leq \textsc{PPS}(\I_2)   d, &\text{as $\textsc{PPS}(\I_2) \le \textsc{PPS}(\I_1)$,} \\
    &\leq 1, &\text{as $\textsc{PPS}(\I_2) < \frac{1}{d}$.}
\end{align*}
As the difference in the total payment to the artists in $C'$ from instances $I_1$ and $I_2$ is at most $1$, so the marginal profit of $N' \setminus \{ i \}$ is at least as good as $N'$, which contradicts the minimality of $N'$.
\end{proof}

The above two lemmas prove that $N'$ consists only of users $i \in [t+1, \dots, t+|U|]$ satisfying $\sum_{j \in C'} w_{ij} = d + 1$.
Let $\I_1$ be the the instance that removes all users in $N'$ from $\I$. Note that all users in $\I$ either stream $d+1$ times or stream $\alpha d$ times. As the removed set of users $N'$ contains only users who stream $d+1$ times, so $\textsc{PPS}(I) < \textsc{PPS}(I_1)$.

Let $L = \sum_{i \in N\setminus N'}\sum_{j \in C'} w_{ij}$. All artist in $[t]$ are streamed $\alpha d$ times, all artists in $[t+1, \dots, t+|V|]$ are streamed at most $|U|$ times, and the artist $N + |V| + 1$ is streamed at most $d |U|$ times. Thus,
\[
    L = \sum_{i \in N\setminus N'}\sum_{j \in C'} w_{ij} \leq |C'| \max_{j \in C'} \sum_{i \in N\setminus N'} w_{ij} \le |C'| d |U| \le d |U| (\delta + 1).
\]

If $|N'| < \ell$, then
\begin{align*}
    \textsc{PSP}(U) &= \textsc{PPS}(\I)  (L + (d+1)|N'| ) - \textsc{PPS}(\I_1)  L - |N'|\\
           &< \textsc{PPS}(\I)  (d+1) |N'| -|N'|, &\text{ as $\textsc{PPS}(\I) < \textsc{PPS}(\I_1)$,}\\
           & \leq \frac{d+1}{d}|N'| - |N'|, &\text{ as $\textsc{PPS}(\I) \leq \frac{1}{d}$,} \\
           & \le \frac{\ell - 1}{d}, &\text{ as $|N'| < \ell$.}
\end{align*}

If $|N'| \geq \ell$, then
\begin{align*}
    \textsc{PSP}(U) &= \textsc{PPS}(\I)  (L + (d+1)|N'|) - \textsc{PPS}(\I_1)  L - |N'|\\
            &=\textsc{PPS}(\I)  (d+1)|N'| - |N'| - (\textsc{PPS}(\I_1) - \textsc{PPS}(\I)) L\\
            &\geq \left( \frac{1}{d} - \epsilon \right) (d+1)|N'| - |N'| - \epsilon L, & \text{ as $\frac{1}{d} - \epsilon \le \textsc{PPS}(\I) \le \frac{1}{d}$ and $\textsc{PPS}(\I_1) \leq \frac{1}{d}$,} \\
            &\geq \frac{|N'|}{d} - \epsilon(L + |N'|) \\
            &\geq \frac{\ell}{d} - \epsilon |U| (d (\delta +1) + 1), &\text{as $|N'| \leq \ell$ and $|N'| \leq |U|$,}\\
            &\ge \frac{\ell - 1}{d}, &\text{by our choice of $\epsilon$.}
\end{align*}

Thus, we have shown that there is a $C' \subseteq C$ such that $|C'| \leq \delta + 1 = k$ and $\textsc{PSP}(C') \ge \frac{\ell - 1}{d} = \gamma$
if and only if there are users $N' \subseteq [t+1, \dots, t+|U|]$ such that $|N'| \geq \ell$ and $\sum_{j \in C'} w_{ij} = d + 1$ for all $i \in N'$.

We claim that the final artist $t + |V| + 1$ is in $C'$.
Notice that the streams of the users in $N \setminus [t] = [t+1, \dots, t+|U|]$ for the artists in $[t+1, \dots, t+|V|]$ have one-to-one correspondence with the edges of the graph, by construction. Therefore, for any user $i \in N \setminus [t]$, the total streams for the artists in $[t+1, \dots, t+|V|]$ is at most the maximum degree $d$ of the graph, i.e., $\sum_{j \in [t+1, \dots, t+|V|]} w_{ij} \le d$ for all $i \in N \setminus [t]$, which implies that $\sum_{j \in [t+1, \dots, t+|V|]} w_{ij} \le d$ for all $i \in N'$ because $N' \subseteq N \setminus [t]$. Further, users in $N \setminus [t]$ do not listen to the first $t$ artists. Therefore, as $\sum_{j \in C'} w_{ij} = d + 1$ for all $i \in N'$, we must have the final artist $t + |V| + 1 \in C'$. This also implies that $|C' \cap [t+1, \dots, t+|V|]| \le |C'|-1 \le \delta$.

Let $S \subseteq U$ be the set that corresponds to $N'$. It is clear that $N(S)$ is a subset of the nodes in $V$ that correspond to $C'$. We note that $|S| = |N'| \geq \ell$ and $|N(S)| = |C' \cap [t+1, \dots, t+|V|]| \leq \delta$. Thus, there is a straightforward bijection between the sets $(N', C')$ such that  $|N'| \geq \ell$ and $|C'| \leq \delta + 1$ and the sets $(S, N(S))$ such that $|S| \geq \ell$ and $|N(S)| \leq \delta$.

\subsection{Proof of \cref{thm:user-additive-implies-monotone-fp-bp}}

A user-additive rule is user-addition monotone.
Let $\I$ be any instance, $\I^{n+1}$ an instance with the addition of a user $n+1$ and arbitrary engagement profile, and $\I_{n+1}$ the instance containing only user $n+1$.
As $\phi$ is user-additive, $\phi_{\I^{n+1}}(c) - \phi_{\I}(c) = \phi_{\I_{n+1}}(c) \geq 0$. By \cref{prop:user-addition-monotone-implies-fp-bp}, it is also fraud-proof and bribery-proof.

\subsection{Proof of \Cref{thm:rules_userprop_strategy}}
We will prove each property separately. Note that the fact that \userprop{} fails strong Sybil-proofness follows from \Cref{thm:strong_sybil_proofness_characterisation}.

We first show that \userprop is user-additive, which will be useful in proving it is also fraud-proof and bribery-proof.
\paragraph{\userprop{} is user-additive.}

This follows immediately from the definition. For any instance $\I = (N, C, \vec{w})$, let $\I^{n+1}$ an instance with a profile $\vec{w}_{n+1}$ appended to $\I$.
Then, for all artists $j$, $\phi_{\I^{n+1}}(j) - \phi_\I(j) = \alpha \frac{w_{n+1, j}}{\sum_{k \in C}w_{n+1, j}}$ which is exactly the payoff of user $j$ in a single user instance with only user $n+1$.

\paragraph{\userprop{} is user-addition monotone, fraud-proof and bribery-proof.}
This claim is just an application of \cref{thm:user-additive-implies-monotone-fp-bp}.

\paragraph{\userprop is Sybil-proof.}
Consider any two instances $\I = (N, C, \mathbf{w})$ and $\I' = (N, C', \mathbf{w}')$ such that $C \subseteq C'$.
Suppose for any subset of artists $C^* \subseteq C$,
\begin{enumerate}[(i)]
    \item $w_{ij} = w'_{ij}$ for all $i \in N, j \in C^*$, and
    \item $\sum_{j \in C \setminus C^*} w_{ij} = \sum_{j \in C' \setminus C^*} w'_{ij}$ for all $i \in N$,
\end{enumerate}
Then, we get that
    \begin{align*}
        \phi_\I(C\setminus C^*)
        & = \sum_{j\in C\setminus C^*} \sum_{i \in N} \frac{w_{ij}}{\sum_{j' \in C} w_{ij'}} \times \alpha \\
        & =  \sum_{i \in N} \frac{\sum_{j\in C\setminus C^*} w_{ij}}{\sum_{j' \in C} w_{ij'}} \times \alpha \\
        & = \sum_{i \in N} \frac{\sum_{j \in C' \setminus C^*} w'_{ij}}{\sum_{j' \in C} w'_{ij'}} \times \alpha \quad \text{(by (i) and (ii))}\\
        & = \sum_{j \in C' \setminus C^*} \sum_{i \in N} \frac{ w'_{ij}}{\sum_{j' \in C} w_{ij'}} \times \alpha\\
        & = \phi_{\I'}(C' \setminus C^*).
    \end{align*}

\paragraph{\userprop{} fails strong Sybil-proofness}
By \cref{thm:strong_sybil_proofness_characterisation}, only \globalprop{} is strongly Sybil-proof.
Hence, \userprop{} is not strongly Sybil-proof.

\subsection{Proof of \Cref{thm:rules_userprop_fairness}}
We will prove each property separately.

\paragraph{\userprop satisfies no free-ridership.}
    Consider an instance $\I = (N, C, \mathbf{w})$.
    For every $j \in C$ where $\sum_{i \in N} w_{ij} = 0$,
    \begin{equation*}
        \phi_{\I}(j) = \sum_{i \in N} \frac{w_{ij}}{\sum_{j' \in C} w_{ij'}} \times \alpha = 0,
    \end{equation*}
    since we assume $\sum_{j' \in C} w_{ij'} > 0$ for all $i \in N$.

\paragraph{\userprop{} is engagement monotone.}
Consider any two instances $\I = (N, C, \mathbf{w})$ and $\I' = (N, C, \mathbf{w}')$ whereby for some $j^* \in C$, we have that (i) $w_{ij^*} \leq w'_{ij^*}$ for all $i \in N$, and (ii) $w_{ij} \geq w'_{ij}$ for all $i \in N$ and $j \in C \setminus \{j^*\}$.

    Now, consider any $i \in N$. Since
    \begin{equation*}
        w'_{ij^*} \geq w_{ij^*}
        \quad \text{and} \quad
        \sum_{j \in C \setminus \{j^*\}} w_{ij} \geq \sum_{j \in C \setminus \{j^*\}} w'_{ij},
    \end{equation*}
    we get that
    \begin{equation*}
        w'_{ij^*} \cdot \sum_{j \in C \setminus \{j^*\}} w_{ij}
        \geq
        w_{ij^*} \cdot \sum_{j \in C \setminus \{j^*\}} w'_{ij}.
    \end{equation*}
    Adding $w'_{ij^*} \cdot w_{ij^*}$ to both sides of the equation, we can factorize the expressions on each side to obtain
    \begin{equation}
        w'_{ij^*} \cdot \left( \sum_{j \in C \setminus \{j^*\}} + w_{ij^*} \right) \geq w_{ij^*} \cdot \left( \sum_{j \in C \setminus \{j^*\}} + w'_{ij^*} \right).
    \end{equation}
    Algebraic manipulation (note that by our model assumption, for each $i \in N$, $\sum_{j' \in C} w_{ij'} > 0$ and $\sum_{j' \in C} w'_{ij'} > 0$) gives us
    \begin{equation*}
        \frac{w_{ij^*}}{\sum_{j \in C \setminus \{j^*\}} w_{ij}} \leq \frac{w'_{ij^*}}{\sum_{j \in C \setminus \{j^*\}} w'_{ij}},
    \end{equation*}
    which simplifies to
    \begin{equation*}
        \frac{w_{ij^*}}{\sum_{j' \in C} w_{ij'}} \leq \frac{w'_{ij^*}}{\sum_{j' \in C} w'_{ij'}}.
    \end{equation*}
    Taking the sum over all users $i \in N$ on both sides, we have that
    \begin{equation*}
        \phi_{\I}(j^*) = \sum_{i \in N} \frac{w_{ij^*}}{\sum_{j' \in C} w_{ij'}} \times \alpha n \leq \sum_{i \in N}
 \frac{w'_{ij^*}}{\sum_{j' \in C} w'_{ij'}} \times \alpha n= \phi_{\I'}(j^*). \qedhere
    \end{equation*}

\paragraph{\userprop{} fails Pigou-Dalton consistency.}
Consider an instance $\I$ with two users and two artists.
Let $\vec{w}_i = (1, 2)$ and $\vec{w}_2= (9, 0)$. Then $\phi_\I(2) = \frac{2}{3}\alpha$.
Suppose instead we consider $\I'$, with $\vec{w}_1' = (1, 1)$ and $\vec{w}_2' = (9, 1)$.
Then, $\I'$ is a Pigou-Dalton improvement on $\I$ as engagement is transferred from a user with higher engagement to a user with a lower engagement.
But, $\phi_{\I'}(2) = \frac{3}{5}\alpha < \phi_\I(2)$ contradicting Pigou-Dalton consistency.

\subsection{Proof of \Cref{thm:rules_usereq_strategy}}

We first show that \userprop is user-additive, which will be useful in proving it is also fraud-proof and bribery-proof.

\paragraph{\usereq is user-additive.}
This follows immediately from the definition. For any instance $\I = (N, C, \vec{w})$, let $\I^{n+1}$ an instance with a profile $\vec{w}_{n+1}$ appended to $\I$.
Then, for all artists $j$, $\phi_{\I^{n+1}}(j) - \phi_\I(j) = \frac{\mathbf{1}_{w_{ij}>0}}{|\{j' \in C : w_{ij'} > 0\}|} \times \alpha$ which is exactly the payoff of artist $j$ in a single user instance with only user $n+1$.

\paragraph{\usereq is user-addition monotone, fraud-proof and bribery-proof.}
As \usereq is user-additive, by \cref{thm:user-additive-implies-monotone-fp-bp}, we have that \usereq is user-addition monotone, fraud-proof and bribery-proof.

\paragraph{\usereq fails Sybil-proofness.}
Consider an instance with one user and two artists, $C = \{1, 2\}$.
Suppose $\vec{w}_1 = (1, 1)$, then $\phi_\I(1) = \half \alpha$.
Suppose instead we consider splitting artist 2 to artists $2'$ and $3'$, with $C' = \{1, 2', 3'\}$. If $\vec{w}_1' = (1, \half, \half)$, \usereq will assign payoff of $\third\alpha$ to each artist, and so the combined payoff of $2'$ and $3'$ in instance $\I'$ is greater than that in $\I$, contradicting Sybil-proofness.

\subsection{Proof of \Cref{thm:rules_usereq_fairness}}
We will prove each property separately.
\paragraph{\usereq satisfies no free-ridership.}
    Consider an instance $\I = (N, C, \mathbf{w})$.
    For every $j \in C$ where $\sum_{i \in N} w_{ij} = 0$,
    \begin{equation*}
        \phi_{\I}(j) = \sum_{i \in N} \frac{\mathbf{1}_{w_{ij}>0}}{|\{j' \in C : w_{ij'} > 0\}|} \times \alpha = 0,
    \end{equation*}
    since we assume $\sum_{j' \in C} w_{ij'} > 0$ for all $i \in N$, and so ${|\{j' \in C : w_{ij'} > 0\}|} > 0$ for all $i \in N$.

\paragraph{\usereq{} is engagement monotone.}
Consider any two instances $\I = (N, C, \mathbf{w})$ and $\I' = (N, C, \mathbf{w}')$ whereby for some $j^* \in C$, we have that (i) $w_{ij^*} \leq w'_{ij^*}$ for all $i \in N$, and (ii) $w_{ij} \geq w'_{ij}$ for all $i \in N$ and $j \in C \setminus \{j^*\}$.

    Consider any $i \in N$. If $w_{ij^*} = 0$, then we trivially get that
    \begin{equation*}
        \phi_{\I}(j^*) = \frac{\mathbf{1}_{w_{ij^*} > 0}}{|\{j' \in C : w_{ij'} > 0 \}|} = 0 \leq \frac{\mathbf{1}_{w'_{ij^*} > 0}}{|\{j' \in C : w'_{ij'} > 0 \}|} = \phi_{\I'}(j^*).
    \end{equation*}
    Note that by our model assumption, $\sum_{j' \in C} w_{ij'} > 0$ and $\sum_{j' \in C} w'_{ij'} > 0$, and thus the fractions are well-defined.
    If $w_{ij^*} > 0$, then $w'_{ij^*} \geq w_{ij^*} > 0$, by (i). Together with (ii), this means that
    \begin{equation*}
        |\{j' \in C : w_{ij'} > 0 \}| \geq |\{j' \in C : w'_{ij'} > 0\}| > 0.
    \end{equation*}
    Then, taking the reciprocal, we get that
    \begin{equation*}
        \frac{1}{|\{j' \in C : w_{ij'} > 0 \}|} \leq \frac{1}{|\{j' \in C : w'_{ij'} > 0\}|}.
    \end{equation*}
    Since $\mathbf{1}_{w_{ij^*} > 0} = \mathbf{1}_{w'_{ij^*} > 0} = 1$, taking the sum over all $i \in N$, we get that
    \begin{equation*}
        \phi_{\I}(j^*) = \sum_{i \in N} \frac{\mathbf{1}_{w_{ij^*} > 0}}{|\{j' \in C : w_{ij'} > 0 \}|} \times \alpha \leq \sum_{i \in N} \frac{\mathbf{1}_{w'_{ij^*} > 0}}{|\{j' \in C : w'_{ij'} > 0 \}|}  \times \alpha  = \phi_{\I'}(j^*). \qedhere
    \end{equation*}

\paragraph{\usereq{} is Pigou-Dalton consistent.}
Consider any two instances $\I = (N, C, \mathbf{w})$ and $\I' = (N, C, \mathbf{w}')$ whereby there exists some $i,i' \in N$ and $j \in C$ such that
\begin{enumerate}[(i)]
    \item $w'_{ij} = w_{ij} - \delta$ (where $\delta > 0$ and $w_{ij} - \delta > 0$);
    \item $w'_{i'j} = w_{i'j} + \delta$ and $w_{i'j} \leq w_{ij}$; and
    \item $w_{kj'} = w'_{kj'}$ for all $k \in N$ and $j' \in C \setminus \{j\}$, and $w_{kj} = w'_{kj}$ for all $k \in N \setminus \{i,i'\}$.
\end{enumerate}
Then, since $w_{ij} > \delta > 0$ (by (i)), this implies $w'_{ij} = w_{ij} - \delta > 0$, giving us
\begin{equation} \label{eqn:usereq_pd_1}
    \mathbf{1}_{w_{ij} > 0} = \mathbf{1}_{w'_{ij} > 0} = 1.
\end{equation}
Also, since $w'_{i'j} > \delta$, we get that
\begin{equation} \label{eqn:usereq_pd_2}
    \mathbf{1}_{w_{i'j} > 0} \leq 1 = \mathbf{1}_{w'_{i'j} > 0}.
\end{equation}

Then, a direct implication from (\ref{eqn:usereq_pd_1}) is
\begin{align*}
    \frac{\mathbf{1}_{w_{ij}>0}}{|\{j' \in C : w_{ij'} > 0\}|} & = \frac{\mathbf{1}_{w_{ij}>0}}{|\{j' \in C \setminus \{j\}: w_{ij'} > 0\}| + \mathbf{1}_{w_{ij}>0}} \\
    & =  \frac{\mathbf{1}_{w'_{ij}>0}}{|\{j' \in C \setminus \{j\}: w_{ij'} > 0\}| + \mathbf{1}_{w'_{ij}>0}} \\
    & = \frac{\mathbf{1}_{w'_{ij}>0}}{|\{j' \in C: w_{ij'} > 0\}|}.
\end{align*}
Moreover, we also get that
\begin{align*}
    \frac{\mathbf{1}_{w_{i'j}>0}}{|\{j' \in C: w_{i'j'} > 0\}|}
    & = \frac{\mathbf{1}_{w_{i'j}>0}}{|\{j' \in C \setminus \{j\}: w_{i'j'} > 0\}| + \mathbf{1}_{w_{i'j}>0}} \\
    & = 1 - \frac{|\{j' \in C \setminus \{j\}: w_{i'j' > 0}\}|}{|\{j' \in C \setminus \{j\}: w_{i'j'} > 0\}| + \mathbf{1}_{w_{i'j}>0}} \\
    & \leq 1 - \frac{|\{j' \in C \setminus \{j\}: w_{i'j' > 0}\}|}{|\{j' \in C \setminus \{j\}: w_{i'j'} > 0\}| + \mathbf{1}_{w'_{i'j} > 0}} \quad \text{(by (\ref{eqn:usereq_pd_2}))}\\
    & = 1 - \frac{|\{j' \in C \setminus \{j\}: w'_{i'j' > 0}\}|}{|\{j' \in C \setminus \{j\}: w'_{i'j'} > 0\}| + \mathbf{1}_{w'_{i'j} > 0}} \quad \text{(by (iii))}\\
    & = \frac{\mathbf{1}_{w'_{i'j}>0}}{|\{j' \in C \setminus \{j\}: w'_{i'j'} > 0\}| + \mathbf{1}_{w'_{i'j}>0}} \\
    & = \frac{\mathbf{1}_{w'_{i'j}>0}}{|\{j' \in C: w'_{i'j'}> 0\}|}.
\end{align*}

Utilizing the two implications obtained above, together with (iii), we get that

\begin{align*}
    \phi_\I(j) & = \sum_{k \in N} \frac{\mathbf{1}_{w_{kj}>0}}{|\{j' \in C : w_{kj'} > 0\}|} \times \alpha \\
    & = \alpha \times \left(  \frac{\mathbf{1}_{w_{ij}>0}}{|\{j' \in C : w_{ij'} > 0\}|} + \frac{\mathbf{1}_{w_{i'j}>0}}{|\{j' \in C : w_{i'j'} > 0\}|} + \sum_{k \in N \setminus \{i,i'\}} \frac{\mathbf{1}_{w_{kj}>0}}{|\{j' \in C : w_{kj'} > 0\}|}\right)\\
    & \leq \alpha \times \left(  \frac{\mathbf{1}_{w'_{ij}>0}}{|\{j' \in C : w'_{ij'} > 0\}|} + \frac{\mathbf{1}_{w'_{i'j}>0}}{|\{j' \in C : w'_{i'j'} > 0\}|} + \sum_{k \in N \setminus \{i,i'\}} \frac{\mathbf{1}_{w'_{kj}>0}}{|\{j' \in C : w'_{kj'} > 0\}|}\right)\\
    & = \sum_{i \in N} \frac{\mathbf{1}_{w'_{ij}>0}}{|\{j' \in C : w'_{ij'} > 0\}|} \times \alpha \\
    & = \phi_{\I'}(j),
\end{align*}
as desired.

\section{Connections to Portioning}
\label{sec:portioning-appendix}

We first formally define a \emph{portioning instance} and \emph{portioning rule}.
\begin{definition}[Portioning Instance]
    A \emph{portioning instance} is an instance $\I = (N, C, \vec{w})$ such that for all $i \in N$, $\norm{\vec{w}_i}_1 = 1$.
\end{definition}

\begin{definition}[Portioning Rule]
    A \emph{portioning rule} is a function $\psi$ that maps each \emph{portioning instance} $\I$ to an $m$-valued vector $ (\psi_{\I}(1), \dots, \psi_{\I}(m))$. Each $\psi_\I(j) \geq 0$ and we require additionally that $\sum_{j \in C} \psi_\I(j) = 1$.
\end{definition}

Because of this relationship, we can generate payment rules by normalizing the engagement vectors and using existing portioning mechanisms.
So, for an instance $\I = (N, C, (w_{ij}))$ we can construct a portioning instance $\I^* = \left(N, C, \left(\frac{w_{ij}}{\norm{\vec{w}_i}_1}\right)\right)$ where $\norm{\vec{w}_i}$ is the $\ell_1$ norm, $\norm{\vec{w}_i} = \sum_{j \in C} w_{ij}$. For a portioning rule $\psi$, we construct a payment rule $\phi$ such that for all artists $j$, the payment is given by the portioning rule $\phi_\I(j) = \psi_{\I^*}(j) \times n \alpha$.

Major portioning rules are cataloged in \citet{elkind2023portioning}.
One broad category of portioning rules are \emph{coordinate-wise} rules.
We can construct these from a function that aggregates the engagement of each artist and then normalize it.

\begin{definition}
    Given a family of functions $f^n \colon (\mathbb{R}_{\geq 0})^n \rightarrow (\mathbb{R}_{\geq 0})$ we can construct a coordinate-wise portioning rule such that the payoff to an artist $j$ is $\psi_\I(j) = \frac{f^n(w_{1j}, w_{2j}, \ldots, w_{nj})}{\sum_{k \in C} f^n(w_{1k}, w_{2k}, \ldots, w_{nk})}$.
\end{definition}

The functions mentioned in \citet{elkind2023portioning} aggregate preferences based on the coordinate-wise \emph{average}, the \emph{maximum}, the \emph{minimum}, the \emph{median} and the \emph{geometric mean}.
From these portioning rules, we can construct analogous payment rules \textsc{Avg}, \textsc{Max}, \textsc{Min}, \textsc{Med} and \textsc{Geo} respectively. We then obtain the following results.

\begin{theorem}
    \textsc{Avg} is equivalent to \userprop. As such it satisfies fraud-proofness, bribery-proofness and Sybil-proofness.
\end{theorem}
\begin{proof}
    Given a problem instance $\I = (N, C, \vec{w})$ with unnormalized $\vec{w}$, \textsc{Avg} will assign artist $j$ a payoff
    $n\alpha \times \frac{\sum_{i \in N}\frac{w_{ij}}{\norm{\vec{w}_i}_1} \frac{1}{n}}{\sum_{k\in C} \sum_{i \in N}\frac{w_{ik}}{\norm{\vec{w}_i}_1} \frac{1}{n}}
    = n\alpha \times
    \frac{\sum_{i \in N}\frac{w_{ij}}{\norm{\vec{w}_i}_1} }{ \sum_{i \in N} \sum_{k\in C}\frac{w_{ik}}{\norm{\vec{w}_i}_1}}$.
    But note that the denominator simplifies to $n$ giving payoff to each artist $j$ equal to $\alpha \sum_{i \in N}\frac{w_{ij}}{\norm{\vec{w}_i}_1}$, which is identical to \userprop.
\end{proof}

\begin{theorem}
    \label{thm:portioning-fraud}
    Rules \Max, \Min, \Geo, \Med, \textsc{Util}, \textsc{Egal} and \textsc{IndependentMarkets} fail fraud-proofness, bribery-proofness and Sybil-proofness for all $\alpha \in (0, 1]$.
\end{theorem}

The eighth rule, \textsc{Avg}, assigns payout proportional to the average engagement of an artist.
This is equivalent to the rule \userprop{}.
The strong axiomatic guarantees of \textsc{Avg} in the portioning setting add an extra layer of support towards \userprop{}.
Conversely, our results that \textsc{Avg} satisfies fraud-proofness and bribery-proofness in our expanded setting add an extra layer of support towards \textsc{Avg}.

To simplify our analysis, we will prove the \Cref{thm:portioning-fraud} using four separate results as follows.

\begin{lemma}
    Coordinate-wise rules \textsc{Max}, \textsc{Min}, \textsc{Med}, \textsc{Geo} fail fraud-proofness and bribery-proofness for all $\alpha \in (0, 1]$, even if there are only two artists.
\end{lemma}
\begin{proof}
    We prove that the rules fail fraud-proofness, the counterexamples for bribery-proofness are very similar.
    For \textsc{Max}, let $n = \ceil*{\frac{6}{\alpha}} + 1$. Let $\vec{w}_i = (\half, \half)$, so that each artist receives a payment of $\frac{n\alpha}{2}$.
    If an adversary in support of $1$ adds $\vec{w}_{n+1} = (1, 0)$ then the payment to $1$ is $\frac{2(n+1)\alpha}{3}$. But, $\frac{2(n+1)\alpha}{3} - \frac{n \alpha}{2}=\frac{4(n+1)\alpha - 3n\alpha}{6} = \frac{n\alpha + \alpha}{6}$. But $n > \frac{6}{\alpha}$ so that the benefit from fraud is greater than 1.

    For \textsc{Min}, let $n = 2 \ceil{\frac{1}{\alpha}}$, $C= \{1, 2\}$ and for all $i \in N$, $\vec{w}_i = (\frac{1}{2}, \frac{1}{2})$, so that each user receives a payoff of $\frac{n\alpha}{2}$.
    Suppose we construct instance $\I'$ by adding profile $\vec{w}_{n+1} = (1, 0)$.
    Then, $\phi_{\I'}(1) = (n+1)\alpha$ and $\phi_{\I'}(1)-\phi_{\I}(1) = (n+1)\alpha - \frac{n\alpha}{2} = \frac{(n+2)\alpha}{2} > 1$ by $n \geq \frac{2}{\alpha}$.

    For \Geo, we can reuse the counterexample for \Min.

    For \textsc{Med}, let $n = \ceil*{\frac{2}{\alpha}}$ if odd or $\ceil*{\frac{2}{\alpha}} + 1$ otherwise and $n = 2k + 1$ for a natural number $k$.
    Then for $i \leq k$, $\vec{w}_i = (1, 0)$ and for $k + 1 \leq i \leq 2k + 1$ let $\vec{w}_i = (0, 1)$.
    Then $\phi_\I(1) = 0$.
    Adding in profile $\vec{w}_{n+1} = (1, 0)$ means $\phi_{\I'}(1) = \frac{(n+1)\alpha}{2} > 1$ by construction.
\end{proof}

Another class of rules focuses on welfare maximization.
For a portioning rule $\psi$ we can measure the disutility of a user $i$ as the $\ell_1$-difference between their engagement and the output payment profile, $d_\I(i) = \sum_{j \in C} \abs{\psi_\I(j) - w_{ij}}$, the user's welfare is then $-d_\I(i)$.
Rule \textsc{Util} maximizes utilitarian welfare $-\sum_{i\in N} d_\I(i)$ and \textsc{Egal} maximizes egalitarian welfare $\min_{i \in N}(-d_\I(i))$.
Ties are broken in favour of the maximum entropy distribution in the case of \textsc{Util}.
For \textsc{Egal}, we break ties in a \emph{leximin} manner, however, our counterexamples do not rely on the tie-breaking method.

\begin{lemma}
    \textsc{Util} and \textsc{Egal} fail fraud-proofness and bribery-proofness for all $\alpha \in (0, 1]$.
\end{lemma}
\begin{proof}
    We prove that the rules fail fraud-proofness, the counterexamples can be slightly modified to also prove bribery-proofness.
    For \textsc{Util}, consider $n = 2k + 1$ and $C=\{1, 2\}$, with $i \leq k + 1$ submitting $\vec{w}_i = (1, 0)$ and $i > k$ submitting $\vec{w}_i = (1, 0)$, then \textsc{Util} will allocate the entire resource to artist $1$ giving payoff $\phi_\I(2) = 0$.
    If a new user is added with $\vec{w}_{n+1} = (0, 1)$ then $\phi_\I(2) = \frac{n\alpha}{2} > 1$ for large enough $n$.

    For \textsc{Egal}, let $C = \{1, 2\}$ and for all $i$, $\vec{w}_i=(\frac{1}{2}, \half)$.
    Then $\phi_\I(1) = \frac{n\alpha}{2}$.
    If we add a profile $(1, 0)$ then to minimize disutility, $\phi_{\I'}(1) = \frac{3}{4}n\alpha$, such that $\phi_{\I'}(1) - \phi_\I(1) = \frac{1}{4}n\alpha > 1$ for large $n$.
\end{proof}

The more sophisticated \emph{independent markets} rule was recently introduced in \citet{freeman2021truthfulbudget}.
This rule is strategy-proof and in some precise sense proportional.
For an instance with $n$ users, the rule constructs $n+1$ phantom values.
Each artist $j$ receives the median of $\{w_{ij} \mid i \in N\}$ and the $n+1$ phantom values.
To compute these phantom values the rule uses functions $f_0, \ldots, f_n \colon [0, 1] \rightarrow [0,1]$ with $f_k(t) = \min(kt, 1)$.
The rule then uses $t^*$ such that the payoff to each artist is $1$, i.e., $\sum_{j \in C}\med(w_{1j}, \ldots, w_{nj}, f_0(t^*), \ldots,f_n(t^*)  ) = 1$. Unfortunately, despite it's sophistication the rule fails to be fraud-proof.

\begin{lemma}
    The \textsc{IndependentMarkets} rule fails to be fraud-proof, bribery-proof or Sybil-proof for all $\alpha \in (0, 1]$.
\end{lemma}

\begin{proof}
    For a number of users $n$, construct an instance $\I_n = (\{1, \ldots, n\}, \{1, \ldots, n+1\}, \vec{w})$, with $w_{i1}=1$ and for artist $j$ with $j \neq 1$, $w_{ij} = 0$.
    Then, $\phi_{\I^n}(1) = n\alpha$ as the users unanimously assign their payoff to user $1$.
    Now, suppose we construct instance $\I_n'$ by adding a user profile $\vec{w}_{n+1} = (0, \frac{1}{n}, \ldots, \frac{1}{n})$.
    Then, there are $n+2$ phantom values generated by the independent markets rule and so each player will be assigned the $n+2$'nd highest value among the phantom and real values.
    For player $1$ that will be the second largest phantom value $t^*n$ and for players $i > 1$ it will be the second lowest phantom value which is $t^*$.
    Given the constraint $nt^* + \sum_{i = 2}^{n + 1} t^* = 1$, we get that $t^* = \frac{1}{2n}$.
    So, the total payoff artists $2, \ldots, n+1$ receive is $tn (n+1)\alpha = \frac{(n+1)a}{2}$.
    So, for $\Chat = C \setminus \{1\}$, $\phi_{\I_n'}(\Chat) - \phi_\I(\Chat) = \frac{(n+1)\alpha}{2} > 1$ for large enough $n$.

    Similarly for bribery-proofness, given an instance $\I_n$, we can construct $\I'$ by setting the profile $\vec{w}_n$ to $(0, \frac{1}{n}, \ldots,\frac{1}{n})$. By the above analysis this generates revenue of $\frac{n\alpha}{2}$ which is greater than $1$ for $n > \frac{2}{\alpha}$.

    For Sybil-proofness, construct an instance $\I = (\{1, \ldots, n + 1\}, \{1, 2\}, \vec{w}\}$ with $\vec{w}_i = (1,0)$ for $i \leq n$ and $\vec{w}_{n+1} = (0, 1)$.
    Then the value users $1, 2$ will be assigned by the independent markets rule is $t^*n$ and $t^*$ respectively.
    As such $\phi_\I(1) = n\alpha$ and $\phi_\I(2) = \alpha$.
    However, from our example in fraud-proofness, we can split user $2$ to users $2', 3', \ldots, n+1'$ and distribute the engagement of user $n+1$ equally.
    This would give a payoff of $\frac{(n+1)\alpha}{2}$ to the Sybil artists which is greater than $\alpha$ for $n > 2$.
\end{proof}

\begin{theorem}
    Rules \Max, \Min, \Geo, \Med, \Util, \Egal fail Sybil-proofness for all $\alpha \in (0, 1]$.
\end{theorem}
\begin{proof}
    For \Max, consider instances $\I = (\{1, 2, 3\}, \{1,2\}, \vec{w})$ with $\vec{w}_1 = (1, 0)$, $\vec{w}_2 = \vec{w}_3 = (0, 1)$.
    Then $\phi_\I(2) = \frac{3\alpha}{2}$. Suppose construct $\I'$ by splitting user $2$ to user $2', 3'$ and $\vec{w}_1' = (1, 0,0), \vec{w}_2' = (0, 1, 0), \vec{w}_3' = (0, 0, 1)$. Then $\phi_\I(2') + \phi_\I(3') = 2\alpha > \phi_\I(2)$ contradicting Sybil-proofness.

    For \Min, consider instance $\I$ with $N = \{1, 2\}$ and $ C=\{1, 2,3\}$ and $\vec{w}_1 = (\third, 0, \frac{2}{3})$, $\vec{w}_2 = (\third, \frac{2}{3}, 0)$, then for $C' = \{2, 3\}$, $\phi_\I(C') = 0$.
    If instead we construct $\I' = (N, C, \vec{w}')$ with $\vec{w}_1' = \vec{w}_1$, $\vec{w}_2' = \vec{w}_3' = (\third, \third, \third)$, then $\phi_{\I}(C') = 2\alpha > \phi_\I(C')$ and contradicting Sybil-proofness.

    For \Geo, we can reuse the example from \Min.

    For \Med, consider $N = \{1, 2, 3\}$, $C = \{1, 2, 3\}$ and $\vec{w}_1 = (1, 0, 0)$, $\vec{w}_2 = (\half, \half, 0)$ and $\vec{w}_3 = (\half, 0, \half)$. For $C' = \{2, 3\}$, $\phi_\I(C') = 0$.
    Now, consider instead $\vec{w}'$, with $\vec{w}_1' = \vec{w}_1$, $\vec{w}_2' =\vec{w}_3' = (\half, \frac{1}{4}, \frac{1}{4})$.
    Then, $\phi_{\I'}(C') = \frac{3\alpha}{2} > \phi_\I(C') = 0$.

    For \Util, consider $N = \{1, 2, 3\}$, $C = \{1, 2, 3\}$ and $\vec{w}_1 = (1, 0, 0)$, $\vec{w}_2 = (0, 1, 0)$ and $\vec{w}_3 = (0, 0, 1)$, then for $C' = \{2,3\}$, $\phi_\I(C') = 2\alpha$.
    Consider instead instance $\I'$ with $\vec{w}'_1 = \vec{w}_1$, $\vec{w}_2' = \vec{w}_3' = (0, \half, \half)$.
    Then, $\phi_{\I'}(C') = 3\alpha > \phi_\I(C')$.

    For \Egal, consider $N = \{1, 2, 3\}$, $C = \{1, 2, 3\}$ and $\vec{w}_1 = (\third, \third, \third)$, $\vec{w}_2 = (0, \half, \half)$ and $\vec{w}_3 = (0, \half, \half)$, then for $C' = \{2,3\}$, $\phi_\I(C') = \frac{5}{2}\alpha$.
    Consider instead instance $\I'$ with $\vec{w}'_1 = \vec{w}_1$, $\vec{w}_2' = (0, 1, 0)$ and $\vec{w}_3' = (0, 0, 1)$.
    Then, $\phi_{\I'}(C') = 3\alpha > \phi_\I(C')$.
\end{proof}

\section{Omitted Proofs from \cref{sec:novelmech}}
\subsection{Proof of \Cref{prop:pps_fpbp}}
Consider an instance $\I$ with $n > \lceil \frac{2}{\alpha} \rceil$ users and two artists. Let $\sum_{i \in N} w_{i1} = \frac{1}{4k}$ and $\sum_{i \in N} w_{i2} = 1$. Then, if $\textsc{ME}(\phi , \I) \leq k$, then $\phi_\I(1) \leq \frac{n}{4}$. Otherwise, if $\phi_\I(1) > \frac{n}{4}$, then  $\textsc{PPS}(\phi, \I, 1) \geq nk$ and $\textsc{PPS}(\phi, \I, 2) \geq \frac{3n}{4}$. Then, $\textsc{ME}(\phi, \I) \geq 4k/4 > k$.

    Next, we add an additional user $i'$ such that $w_{i'1} = 3k$ and $w_{i'2} = 0$. Let this instance be $\I'$. Then, if $\textsc{ME}(\phi , \I') \leq k$, then $\phi_{\I'}(1) \geq \frac{3(n+1)}{4}$. Otherwise, if $\phi_{\I'}(1) < \frac{3(n+1)}{4}$, then  $\textsc{PPS}(\phi, \I, 1) < (n+1)/4k$ and $\textsc{PPS}(\phi, \I, 2) \geq (n+1)/4$. Then, $\textsc{ME}(\phi, \I) > \frac{(n+1)/4}{ (n+1)/4k} > k$.

    Thus, if $\textsc{ME}(\phi , \I) \leq k$ and $\textsc{ME}(\phi , \I') \leq k$, then $\phi_{\I'}(1) - \phi_{\I}(1) \geq \frac{3(n+1)}{4} - \frac{n}{4} > \frac{n}{2}$. As $n > \lceil \frac{2}{\alpha} \rceil$, then $\phi_{\I'}(1) - \phi_{\I}(1) > 1$ and $\phi$ is not fraud-proof.

    By modifying instance $\I$ and having user $i'$ such that $w_{i'1} = 0$ and $w_{i'2} = \epsilon$, a similar argument shows that $\phi$ is not bribery-proof.
\subsection{Proof of \cref{thm:scaled-up-similar-to-gp}}
Here, we let $\norm{\vec{w}_i}_1 = \sum_{j \in C}w_{ij}$.
\paragraph{For an instance where for all $i$, $\norm{\vec{w}_i}_1 \leq \frac{1}{n\alpha} \sum_{i' \in n} \norm{\vec{w}_{i'}}_1$ \scaledUP and \globalprop{} give the same payoff to each artist.}
If for all $i$, $\norm{\vec{w}_i}_1 \leq \frac{1}{n\alpha} \sum_{i' \in n} \norm{\vec{w}_{i'}}_1$, $\gamma = \frac{n\alpha}{\sum_{i\in N} \norm{\vec{w}_i}_1}$.
From our inequality we have that $\gamma \norm{\vec{w}_i}_1 \leq \frac{\gamma}{n\alpha} \sum_{i \in N} \norm{\vec{w}_i}_1 = 1$ and so in particular $\min(\gamma \norm{\vec{w}_i}_1, 1) = \gamma \norm{\vec{w}_i}_1$. Also, $\sum_{i \in N} \gamma \norm{\vec{w}_i}_1 = n\alpha$, so this is the appropriate $\gamma$.

So, the payoff to each artist  is:
\[\phi_\I(j) = \sum_{i \in N} \gamma \norm{\vec{w}_i}_1 \frac{w_{ij}}{\norm{\vec{w}_i}_1} = \sum_{i \in N} \gamma w_{ij} = n\alpha  \frac{\sum_{i \in N}w_{ij}}{\sum_{i \in N} \norm{\vec{w}_i}_1}.\]

Which is identical to \globalprop{}.
\subsection{Proof of \Cref{thm:rules_scaled_strategy}}

We will prove each property separately.

\paragraph{\scaledUP is bribery-proof.}
Suppose for a contradiction the \scaledUP does not satisfy bribery-proofness.
Then there are instances $\I = (N, C, \vec{w})$ and $\I' = (N, C, \vec{w}')$ with $\vec{w}_i = \vec{w}'_i$ for $i < n$ and $\vec{w}_n \neq \vec{w}'_n$ such that for a $C^+ \subseteq C$, $\phi_{\I'}(C^+) - \phi_\I(C^+) > 1$.
We will prove this result by simplifying the cases we need to consider.
First, note that without loss of generality we can collapse $C^+$ to a single artist.
For any instance $\mathcal{J}$, we can construct an instance $\mathcal{J}^*$ by collapsing artists $C^+$ to a single artist in $\mathcal{J}^*$. Each user $i$ has engagement to a fresh user $c^+$ equal to $\sum_{j \in C^+} w_{ij}$ then in \scaledUP, $\phi_{\mathcal{J}^*}(c^+) = \phi_{\mathcal{J}}(C^+)$. Similarly, for the purposes of this proof we can collapse the complement $C \setminus C^+$ to a single user.
So without loss of generality, it suffices to prove the result for $C = \{1, 2\}$.

Also, suppose ${w}_{n1} > 0$, then setting $w_{n1}$ to $0$ would weakly decrease the payoff of artist $1$ in instance $\I$ and so increase the profit from bribery. So without loss of generality, ${w}_{n1} = 0$ and similarly ${w}_{n2}' = 0$.
By engagement monotonicity, the maximum difference $\phi_{\I'}(1) - \phi_{\I}(1)$ is achieved for profiles $\vec{w}_n = (0, M)$ and $\vec{w}_n = (M, 0)$ for large $M$.

If $n\alpha \leq 1$ then bribery is inherently impossible as the mechanism does not distribute enough payoff to cover a single subscription fee.
If $n\alpha > 1$ then it suffices to consider the minimum $M^*$ such that $\gamma M^*  \geq 1$. Increasing $M$ past $M^*$ does not affect $\gamma$.

But note: $\gamma$ in $\I$ and $\gamma'$ in $\I$ are equal! So, $\phi_{\I'}(1) - \phi_\I(1) = \min(\gamma M, 1) \frac{M}{M} - \min(\gamma M, 1) \frac{0}{M} = 1$.
So, the maximum benefit from bribing is at most 1, proving bribery-proofness of \scaledUP.

\paragraph{\scaledUP is Sybil-proof.}
Consider any two instances $\I = (N, C, \mathbf{w})$ and $\I' = (N, C', \mathbf{w}')$ such that $C \subseteq C'$.
Suppose for any subset of artists $C^* \subseteq C$,
\begin{enumerate}[(i)]
    \item $w_{ij} = w'_{ij}$ for all $i \in N, j \in C^*$, and
    \item $\sum_{j \in C \setminus C^*} w_{ij} = \sum_{j \in C' \setminus C^*} w'_{ij}$ for all $i \in N$,
\end{enumerate}
Let $\gamma$ and $\gamma'$ be constants such that
\begin{equation*}
    \sum_{i \in N}  \min(\gamma \cdot \sum_{j \in C}w_{ij}, 1)= \alpha n \quad \text{and} \quad \sum_{i \in N}  \min(\gamma' \cdot \sum_{j \in C}w'_{ij}, 1)= \alpha n, \quad \text{respectively}.
\end{equation*}
Then, using (i) and (ii), we equivalently get that $\gamma$ and $\gamma'$ are constants such that
\begin{equation*}
    \sum_{i \in N}  \min(\gamma \cdot \sum_{j \in C}w'_{ij}, 1)= \alpha n \quad \text{and} \quad \sum_{i \in N}  \min(\gamma' \cdot \sum_{j \in C}w_{ij}, 1)= \alpha n, \quad \text{respectively}.
\end{equation*}
This means that $\gamma = \gamma'$.
Then, we get that
    \begin{align*}
        \phi_\I(C\setminus C^*)
        & = \sum_{j\in C\setminus C^*} \sum_{i \in N} \min(\gamma \cdot \sum_{j' \in C} w_{ij'}, 1) \times \frac{w_{ij}}{\sum_{j' \in C} w_{ij'}} \\
        & = \sum_{i \in N} \min(\gamma \cdot \sum_{j' \in C} w_{ij'}, 1) \times \frac{\sum_{j\in C\setminus C^*} w_{ij}}{\sum_{j' \in C} w_{ij'}} \\
        & = \sum_{i \in N} \min(\gamma' \cdot \sum_{j' \in C} w'_{ij'}, 1) \times \frac{\sum_{j\in C\setminus C^*} w'_{ij}}{\sum_{j' \in C} w'_{ij'}} \quad \text{(by (i), (ii), and since $\gamma = \gamma'$)}\\
        & = \sum_{j\in C\setminus C^*} \sum_{i \in N} \min(\gamma' \cdot \sum_{j' \in C} w'_{ij'}, 1) \times \frac{w'_{ij}}{\sum_{j' \in C} w'_{ij'}} \\
        & =  \phi_{\I'}(C' \setminus C^*).
    \end{align*}

\paragraph{\scaledUP is fraud-proof.}
Denote $\norm{\vec{w}_i}_1 = \sum_{j \in C} w_{ij}$.

We prove this result by first simplifying the cases needed to consider.
Consider instances $\I = (N, C, \vec{w})$ and $\I' = (N \cup \{n+1\}, C, \vec{w}')$ such that for $i < n+1$, $\vec{w}_i = \vec{w}'_i$ but for some coalition of artists $C^* \subseteq C$, $\phi_{\I'}(C^*) - \phi_{\I}(C^*) > 1$.
Similarly to the proof of bribery-proofness, without loss of generality the coalition $C^*$ contains a single user $m$. In this new instance, $m$ receives engagement from user $i$ equal to $\sum_{j\in C^*}w_{ij}$.

Also, for any vector $\vec{w}_{n+1}$ with fixed $\ell_1$-norm, the payoff to user $m$, $\phi_\I(m)$, is maximized for $w_{n+1, j} = 0$ for $j < \abs{C}$ and $w_{{n+1}, m} = \norm{\vec{w}_{n+1}}_1$.
Fixing $\norm{\vec{w}_{n+1}}_1$ fixes $\gamma$ and to maximize the term $\frac{w_{{n+1}, m}}{\norm{\vec{w}_{n+1}}_1}$, we place all engagement in coordinate $w_{{n+1}, m}$.
So without loss of generality, it suffices to consider $\vec{w}_{n+1}$ only of the form $(0,0,\ldots,0, M)$.

By engagement monotonicity, for $M < M'$, if $\vec{w}_{n + 1} = (0,0,\ldots,0, M)$ is a fraud-proofness violation, so is  $\vec{w}_{n + 1} = (0,0,\ldots,0, M')$.

Now, let $\gamma$ and $\gamma'$ be the parameters produced in instances $\I$ and $\I'$ respectively.
Without loss of generality we consider instances of the form $\vec{w}_{n} = (0,0,\ldots,0, M)$, with the property that $\gamma' M > 1$.
This is possible because we assume that $(n+1) \alpha > 1$, which is a requirement for there to be fraud.
Then:
\begin{align*}
 \phi_{\I'}(m) - \phi_\I(m) &= 1 + \sum_{i = 1}^{n} \left(\min(\gamma' \norm{\vec{w}_i}_1, 1) - \min(\gamma \norm{\vec{w}_i}_1, 1)\right) \times \frac{w_{ij}}{\norm{\vec{w}_i}_1}
\end{align*}

But, $\gamma' \leq \gamma$ because $\gamma' \vec{w}_{n+1} \geq 1$ and so $\alpha (n + 1) = 1 + \sum_{i = 1}^n \min(\gamma' \norm{\vec{w}_i}_1, 1) \implies \sum_{i = 1}^n \min(\gamma' \norm{\vec{w}_i}_1, 1) = n \alpha - 1 + \alpha \leq n \alpha = \sum_{i = 1}^n \min(\gamma \norm{\vec{w}_i}_1, 1)$.
So, $\min(\gamma' \norm{\vec{w}_i}_1, 1) - \min(\gamma \norm{\vec{w}_i}_1, 1) \leq 0$ and so $\phi_{\I'}(m) - \phi_\I(m) \leq 1$, proving fraud-proofness.

\paragraph{\scaledUP fails strong Sybil-proofness}
This follows directly from \Cref{thm:strong_sybil_proofness_characterisation} as the only rule satisfying strong Sybil-proofness is \globalprop{}.

\subsection{Proof of \Cref{thm:rules_scaled_fairness}}
We will prove each property separately.

    \paragraph{\scaledUP satisfies no free-ridership.}
    Consider an instance $\I = (N, C, \mathbf{w})$.
    For every $j \in C$ where $\sum_{i \in N} w_{ij} = 0$,
    \begin{equation*}
        \phi_\I(j) = \sum_{i \in N} \min(\gamma \cdot \sum_{j' \in C} w_{ij'}, 1) \times \frac{w_{ij}}{\sum_{j' \in C} w_{ij'}} = 0,
    \end{equation*}
    since we assume $\sum_{j' \in C} w_{ij'} > 0$ for all $i \in N$.

    \paragraph{\scaledUP is engagement monotone.} Denote $\norm{\vec{w}_i}_1 = \sum_{j \in C} w_{ij}$ and for a specific \scaledUP instance, we write $\alpha_i$ as a shorthand for  $\min(\gamma \cdot \sum_{j' \in C} w_{ij'}, 1)$.

    Consider any two instances $\I = (N, C, \vec{w})$ and $\I' = (N, C, \vec{w}')$ such that for with $i \neq n$ or $j \neq m$, $w_{ij}' = w_{ij}$ but $w'_{ij} > w_{ij}$.
    Let $\gamma, \alpha_1, \ldots, \alpha_n$ and $\gamma', \alpha_1', \ldots, \alpha_n'$ the values computed for instances $\I$ and $\I'$ respectively.
    If $\gamma \norm{\vec{w}_n}_1 \geq 1$, then $\gamma'= \gamma$ and so for $j < m$, $\alpha'_j = \alpha_j$.
    So, we have $\phi_{\I'}(m) - \phi_{\I}(m) = \frac{w_{nm}'}{\norm{\vec{w}_{n}'}} - \frac{w_{nm}}{\norm{\vec{w}_{n}}} \geq 0$.

    Suppose that $\gamma \norm{\vec{w}_n}_1 < 1$. Then we must have $\gamma' < \gamma$.
    So for $i < n$, $\alpha_i' \leq \alpha_i$ and $\alpha_n' \geq \alpha_n$.
    By $n\alpha = \sum_{i \in N} \alpha_i = \sum_{i \in N} \alpha_i'$, $\alpha_n' - \alpha_n = \sum_{i = 0}^{n - 1} \alpha_i - \alpha_i'$.

    Suppose in addition, $\gamma' \norm{\vec{w}_n'}_1 \leq 1$.
    So, artist $m$ loses payoff of at most $\gamma'\norm{\vec{w}_n'}_1 - \gamma\norm{\vec{w}_n}_1$ from a reduction of payment from users $1, \ldots n-1$.
    However, she makes $\gamma'w_{nm}'  - \gamma w_{nm}$ more from the contribution of user $n$.
    But, $(\gamma'w_{nm}'  - \gamma w_{nm}) - (\gamma'\norm{\vec{w}_n'}_1 - \gamma\norm{\vec{w}_n}_1)
    = \gamma(\norm{\vec{w}_n}_1 - w_{nm}) - \gamma'(\norm{\vec{w}_n'}_1 - w_{nm}')  \geq 0$
    as $\gamma \geq \gamma'$ and $\norm{\vec{w}_n}_1 - w_{nm} = \norm{\vec{w}_n'}_1 - w_{nm}'$.

    To prove the case $\gamma' \norm{\vec{w}_n'}_1 > 1$ we can simply consider an intermediate instance $\I''$ such that $w_{nm} < w_{nm}'' < w_{nm}$ and $\gamma'' \norm{\vec{w}''}_1 = 1$. We have proven that the payoff of user $m$ increases from $\I$ to $\I'$ and from $\I'$ to $\I''$ and hence from $\I$ to $\I'$.

    \paragraph{\scaledUP fails strong Sybil proofness.}
    This follows directly from \Cref{thm:strong_sybil_proofness_characterisation} as the only rule satisfying strong Sybil-proofness is \globalprop{}.

    \paragraph{\scaledUP fails Pigou-Dalton consistency for every $\alpha \in (0, 1]$.}

    Denote $\norm{\vec{w}_i}_1 = \sum_{j \in C} w_{ij}$.

    Fix $\alpha \leq 1$. Then, let $n = \ceil{\frac{1}{\alpha}} + 1$ and construct instance $\I = (\{1, 2, \ldots, n\}, \{1, 2\}, \vec{w})$.
    For $i < n$, let $w_{i1} = 1, w_{i2} = 0$.
    Let $w_{i1} = \frac{M}{2}, w_{i2} = \frac{M}{2}$ for $M = \frac{\ceil{\frac{1}{\alpha}}}{n\alpha - 1}$.

    Then $\gamma = \frac{n\alpha - 1}{n - 1}$ as $\gamma \norm{\vec{w}_n}_1 = \frac{n\alpha - 1}{n - 1} \frac{\ceil{\frac{1}{\alpha}}}{n\alpha - 1} = 1$ and so $\sum_{i = 1}^n \min(\gamma \norm{\vec{w}_i}_1, 1) = 1 + \sum_{i = 1}^{n-1} \gamma = n \alpha$.
    So, artist $1$ receives payoff $\phi_\I(1) = \frac{1}{2} + n\alpha - 1$.
    Suppose now we construct instance $\I'$ identical to $\I$, except $w_{11} = \frac{1}{2}$ and $w_{n1} = \frac{M + 1}{2}$.

    Then, $\gamma' = \frac{n\alpha - 1}{n - 1.5}$ and so $\gamma' > \gamma$ and in particular $\gamma'(M + \frac{1}{2}) > 1$.

    So, artist $1$ receives payoff $\phi_{\I'}(1) = \frac{M + 1}{2M + 1} + n\alpha - 1 > \phi_\I(1) = \frac{1}{2} + n\alpha - 1$.
    This proves that for all $\alpha$ there is an instance that violates Pigou-Dalton consistency.

\end{document}